\documentclass[journal]{IEEEtran}

\usepackage{cite}
\usepackage{color}
\ifCLASSINFOpdf
\usepackage[pdftex]{graphicx}
\graphicspath{{../pdf/}{../jpeg/}}
\else
\fi
\usepackage{amsthm,amsmath,mathrsfs,amsfonts,indentfirst}
\usepackage{algorithm}
\usepackage{algpseudocode}
\usepackage{array,graphicx,subfig}
\usepackage{epsfig}

\usepackage{stfloats}

\usepackage{setspace}

\usepackage{lipsum}

\newtheorem{theorem}{Theorem}
\newtheorem{lemma}{Lemma}

\begin{document}
%
\title{Millimeter-Wave NR-U and WiGig Coexistence: Joint User Grouping, Beam Coordination and \\ Power Control}

\author{Xiaoxia Xu, Qimei Chen, ~\IEEEmembership{Member,~IEEE}, Hao Jiang, ~\IEEEmembership{Member,~IEEE}, and Jun Huang
             \thanks{Manuscript received Sep 16, 2020; revised Apr 28, 2021; accepted Sep 05, 2021. This work was supported in part by the National Natural Science Foundation Program of China under Grant 61801333, in part by the Young Elite Scientists Sponsorship Program by CAST 2019, in part by the Natural Science Foundation General Program of Hubei Province under Grant 2020CFB633, in part by the National Natural Science Foundation of China Enterprise Innovation Development Key Project under Grant U19B2004.  (\emph{Corresponding author: Qimei Chen, Hao Jiang})}
             \thanks{X. Xu, Q. Chen, and H. Jiang are with the School of Electronic Information, Wuhan University, Wuhan 430072, China (e-mail: \{xiaoxiaxu, chenqimei, jh\}@whu.edu.cn).}
             \thanks{J. Huang is with the School of Cybersecurity, Northwestern Polytechnical University, Taichang/Xi'an, China. e-mail: huangj@ieee.org.}}

\markboth{IEEE TRANSACTIONS ON WIRELESS COMMUNICATIONS,~Vol.~XX, No.~XX, XXX~2021}%
{Shell \MakeLowercase{\textit{et al.}}: Bare Demo of IEEEtran.cls for IEEE Journals}

\maketitle

\begin{abstract}
Millimeter wave (mmWave) communication is a promising New Radio in Unlicensed (NR-U) technology to meet with the ever-increasing data rate and connectivity requirements in future wireless networks.
However, the development of NR-U networks should consider the coexistence with the incumbent Wireless Gigabit (WiGig) networks.
In this paper, we introduce a novel multiple-input multiple-output non-orthogonal multiple access (MIMO-NOMA) based mmWave NR-U and WiGig coexistence network for uplink transmission.
Our aim for the proposed coexistence network is to maximize the spectral efficiency while ensuring the strict NR-U delay requirement and the WiGig transmission performance in real time environments.
A joint user grouping, hybrid beam coordination and power control strategy is proposed, which is formulated as a Lyapunov optimization based mixed-integer nonlinear programming (MINLP) with unit-modulus and nonconvex coupling constraints.
Hence, we introduce a penalty dual decomposition (PDD) framework, which first transfers the formulated MINLP into a tractable augmented Lagrangian (AL) problem. Thereafter, we integrate both convex-concave procedure (CCCP) and inexact block coordinate update (BCU) methods to approximately decompose the AL problem into multiple nested convex subproblems, which can be iteratively solved under the PDD framework.
Numerical results illustrate the performance improvement ability of the proposed strategy, as well as demonstrating
the effectiveness to guarantee the NR-U traffic delay and WiGig network performance.
\end{abstract}


\IEEEpeerreviewmaketitle

\section{Introduction}
\IEEEPARstart{F}{acing} with the surge of mobile devices and the spectrum crunch of microwave bands, the upcoming fifth generation (5G) wireless networks are expected to consider millimeter wave (mmWave) bands to meet the demands of massive connectivity, high transmission data rate, and low traffic delay \cite{Lopez2014Opportunities}.
However, transmissions at mmWave bands would suffer from significantly high propagation loss and susceptibility to blockage \cite{Rappaport2012mmWave}.
In addition, it is extremely challenging to meet the low hardware cost and energy consumption constraints of mobile users due to the mmWave radio frequency (RF) complexity \cite{Roh2014mmWaveBF, Molisch2017HybridBF}.
Therefore, mmWave communications are historically infeasible.
Recently, 5G New Radio-Unlicensed (NR-U) network has been confirmed to appealingly include unlicensed $60$ GHz mmWave band, which tackles these challenges by mobilizing mmWave into the advanced NR technologies \cite{3GPP2019NRU}. Specifically, the NR-U network equipped with large antenna arrays can compensate the propagation characteristics by providing highly directional transmissions \cite{Roh2014mmWaveBF}. Moreover, hybrid beamforming can also be adopted to reduce the NR-U RF complexity \cite{Molisch2017HybridBF}.

However, the NR-U network is still incapable to support massive connectivity due to the limitation of RF chains.
Currently, non-orthogonal multiple access (NOMA) is considered as a promising multiple access technology for future wireless communications \cite{Liu2017NOMA,Chandra2018mmWaveNOMA}, which aims to simultaneously serve multiple users over the same frequency/time/code at the cost of inter-user interference.
Hence, it is natural to apply the NOMA technique into the NR-U network.
There are several existing works investigating the multiple-input multiple-output (MIMO)-NOMA technique.
The authors in \cite{Ding2017MIMONOMA} introduces two patterns of MIMO-NOMA transmission protocols, where NOMA technique is combined with multiple antennas utilized either for beamforming or spatial multiplexing.
With the aid of random beamforming, the authors in \cite{Cui2018MIMONOMA} investigate a joint user scheduling and power control design.
In \cite{Zhu2019mmWaveNOMA}, the authors divide the users with high-correlated spatial channels into different directional beams, and formulate a hybrid beamforming and power control problem to maximize data rate.
The authors in \cite{Xiao2018mmWaveNOMA} propose a joint beamforming and power control strategy for downlink mmWave NOMA transmissions under a constant modulus constraint due to the analog beamforming structure. Correspondingly, the authors in \cite{Zhu2018UplinkmmWaveNOMA} propose a joint beamforming and power control mechanism for uplink mmWave NOMA transmissions, where a close-to-bound uplink data rate has been achieved.
Furthermore, the authors in \cite{Dai2019HybridPrecodingMIMONOMA} investigate the integration of simultaneous wireless information and power transfer (SWIPT) in mmWave massive MIMO-NOMA systems.
An angle-domain MIMO-NOMA scheme for multi-cell mmWave networks is proposed in \cite{Shao2020MIMONOMA_MultiCell}, where each beam serves one cell-edge user and one cell-center user, and the precoders and decoders are jointly designed to mitigate interference.
However, all of the aforementioned studies have not considered the coexistence issue with other radio access technologies (RATs) over $60$ GHz mmWave band, mainly for the incumbent wireless gigabit (WiGig) networks.

IEEE $802.11$ay \cite{Ghasempour2017802.11ay,802.11ayDraft2017} is an enhanced WiGig protocol building upon IEEE $802.11$ad, which enables multi-user MIMO  (MU-MIMO) beamforming to realize multi-Gigabit-per-second ultra-high-speed communications.
Since both NR-U and WiGig networks adopt narrow directional beams, the conventional listen-before-talk (LBT) mechanisms designed for the network coexistence under Sub-6GHz bands cannot be applied directly \cite{Lagen2019NRBeamBasedCoexistence}.
The authors in \cite{mmWaveLBT2017} propose an omnidirectional LBT (omniLBT) mechanism and a directional LBT (dirLBT) mechanism for the mmWave coexistence network.
However, the omniLBT mechanism may prevent transmission even if the detected signal insignificantly damages the transmissions in the intended direction, which results in overprotection and spatial inefficiency. Meanwhile, the dirLBT mechanism suffers the hidden nodes problem.
A tradeoff between omniLBT and dirLBT mechanisms is investigated in \cite{Lagen2018pairLBT}.
However, both omniLBT and dirLBT mechanisms are performed based on the interference estimation at the transmitters, which is usually impractical.
As a result, the authors in \cite{LAT2019,Lagen2018LBR} respectively investigate a more practical and efficient LBT mechanism that estimates interference at receivers.
Specifically, the authors in \cite{LAT2019} propose a listen-after-talk (LAT) mechanism that directly involves carrier sensing at the receiver side.
However, LAT is not compliant with LBT regulation.
To address these issue, a listen-before-receive (LBR) mechanism is proposed \cite{Lagen2018LBR}, where transmitters trigger the receivers to feedback assisted carrier sensing information to complement LBT.
Nonetheless, due to the lack of centralized control, the spectrum utilization is insufficient under the distributed and uncoordinated LBT based mechanisms.
Particularly, these mechanisms are inapplicable in MIMO-NOMA based networks, since the intra-beam interference cancellation is power dependent.


In this paper, we introduce a novel MIMO-NOMA based uplink transmission coexistence framework for multi-cell NR-U and WiGig networks.
Different from the existing works, the proposed coexistence framework enables NR-U and WiGig devices to simultaneously transmit on the same mmWave band through the coordination of highly-directional beams.
By leveraging the MIMO-NOMA technique and observing WiGig signals, the NR-U network is capable to support higher spectral efficiency and lower traffic latency for intensively connected users, as well as preventing performance loss to incumbent WiGig transmissions.
Therefore, we propose a joint user grouping, hybrid beam coordination and power control strategy to maximize the spectral efficiency while ensuring the NR-U delay requirement and WiGig network performance in real time environments. The objective function is formulated as a Lyapunov optimization based mixed-integer nonlinear programming (MINLP) with unit-modulus and nonconvex coupling constraints.
To deal with the NP-hard problem, we first equivalently transfer the nonconvex MINLP into a tractable augmented Lagrangian (AL) problem under a penalty dual decomposition (PDD) framework.
Thereafter, we approximately decompose the AL problem into several nested convex subproblems through the concave-convex procedure (CCCP) and inexact block coordinate update (BCU) methods, which can be iteratively solved under the PDD framework.


The main contributions of this paper are listed as follows.
\begin{enumerate}
  \item We introduce a novel mmWave NR-U and WiGig coexistence framework for uplink transmission.
  By leveraging MIMO-NOMA technique, the proposed framework can enhance user connections among multiple RATs over the same mmWave frequency band.
  Through the coordination of the highly-directional beams, both intra-RAT and inter-RAT interferences can be suppressed to improve spectrum utilization as well as ensuring harmonious coexistence.

  \item A joint user grouping, hybrid beam coordination and power control strategy has been proposed to maximize the NR-U spectral efficiency while ensuring each NR-U users' data rate and delay requirements, as well as guaranteeing the performance of WiGig networks. The objective function is formulated as a Lyapunov optimization problem, which can asymptotically achieve the optimal solution and dynamically control the tradeoff between NR-U traffic delay and WiGig interference mitigation.

  \item We propose a \emph{PDD-CCCP} method to solve the NP-hard Lyapunov optimization problem.
      The nonconvex MINLP Lyapunov optimization problem is tractably transferred as an AL problem.
       By conjunctively utilizing CCCP and inexact BCU methods, the AL problem is approximately decomposed into several nested convex subproblems, which can be iteratively solved based on the PDD framework.
\end{enumerate}

The rest of this paper is organized as follows. Section II introduces the NR-U and WiGig coexistence framework. Section III describes the system model and Section IV is the problem formulation. In Section V, a joint user grouping, beam coordination and power control strategy is studied. Numerical results are presented in Section VI. Finally, Section VII concludes the whole paper.
For the ease of reference, we summarize the definitions of the acronyms that will be frequently utilized in this work in Table \ref{Tab:01}.

\emph{Notation:} We denote the vectors and matrices by lower and upper boldface symbols, respectively.
$\frac{\partial F}{ \partial x}$ denotes the first partial derivative of function $F$ with respect to $x$. $\mathbb{E}\left[\cdot\right]$ represents the statistical expectation. $(\cdot)^H$ is the Hermitian transpose. $|\cdot|$ and $\|\cdot\|$ denote the absolute value and the Euclidean norm, respectively.

\begin{table}
\centering
\caption{Summary of important acronyms}
\label{Tab:01}
\small
\begin{tabular}{|c|c|}
  \hline
  \textbf{Acronym} & \textbf{Definition} \\\hline
  BI & Beacon interval \\\hline
  BFT & Beamforming training \\\hline
  SLS & Sector level sweep \\\hline
  BTI & Beacon transmission interval \\\hline
  A-BFT &  Association beamforming training\\\hline
  BRP & Beam refinement period \\\hline
  BHI & Beacon header interval \\\hline
  DTI & Data transmission interval \\\hline
  ATI & Announcement transmission interval \\\hline
  DMG & Directional multi-gigabit \\\hline
  SSW & Sector sweep \\\hline
  SSW-FB & Sector sweep feedback \\\hline
  SID & Sector identification\\\hline
  SP & Service period \\\hline
\end{tabular}
\end{table}

\section{MmWave NR-U and WiGig Coexistence Framework}
\begin{figure}[h]
    \centering
    \includegraphics[width=0.5\textwidth]{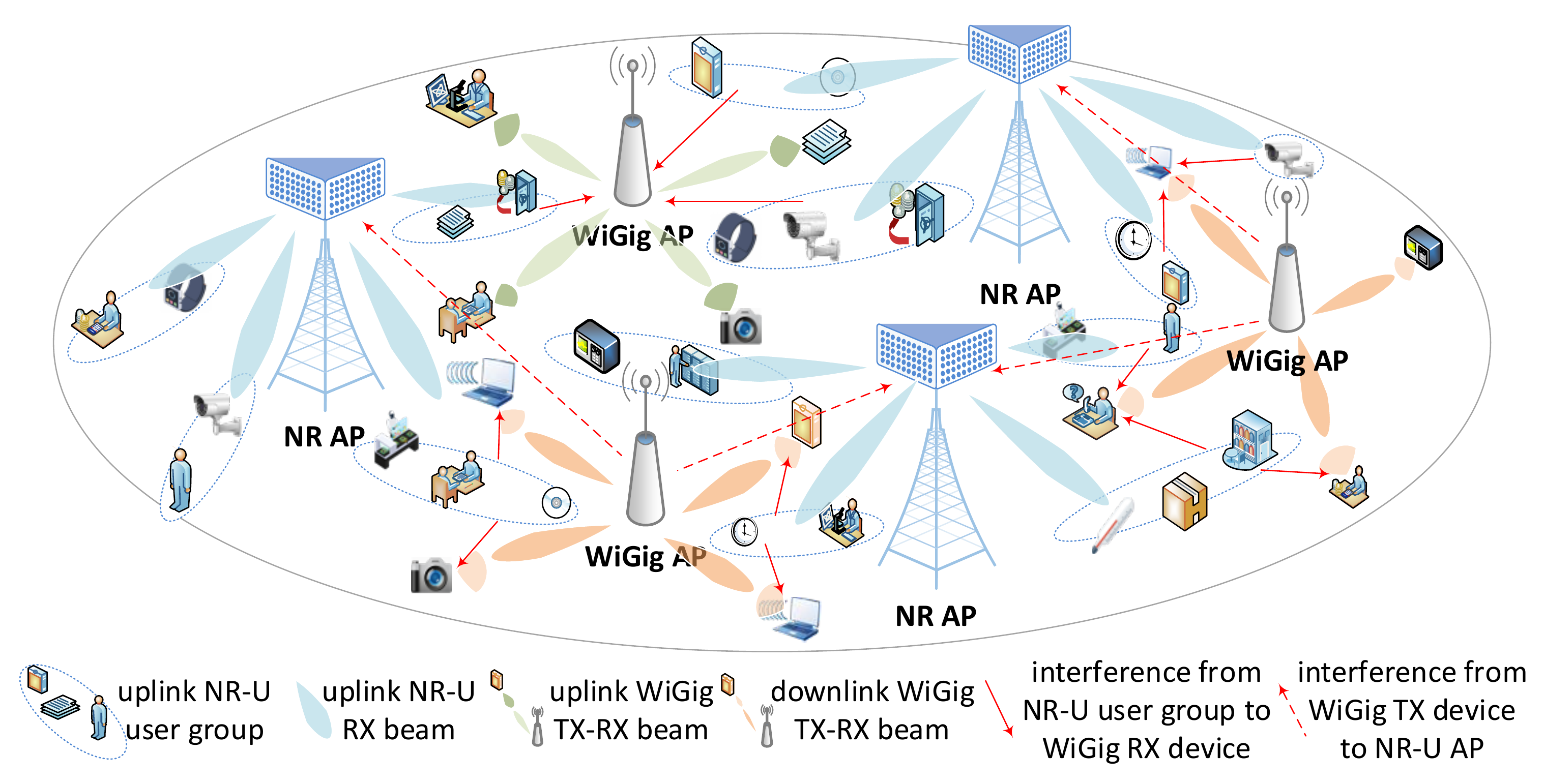}
    \caption{Illustration of the NR-U and WiGig coexistence framework.}
    \label{sysModel}
\end{figure}
As shown in Fig. \ref{sysModel}, the WiGig networks complying with the IEEE 802.11ay standard \cite{802.11ayDraft2017} have an overlapped coverage with the mmWave NR-U networks deployed indoors. In this work, we mainly consider the uplink NR-U transmissions.
Meanwhile, the WiGig networks support both uplink and downlink MU-MIMO transmissions.
We assume both the WiGig APs and NR-U APs apply the hybrid beamforming that allow multiple transmission links simultaneously, and each WiGig user utilizes analog beamforming with single RF chain to reduce hardware cost and complexity\cite{Aldubaikhy2019HBF-PDVG}.
We further assume that each NR-U AP integrates both the NR-U and WiGig interfaces.
Therefore, each NR-U AP can estimate the channel state information (CSI) from both the NR-U users and WiGig devices.
In addition, we assume the NR-U AP has the pre-defined WiGig analog beamforming codebook knowledge in advance.


\begin{figure*}[hbt]
    \centering
    \includegraphics[width=0.8\textwidth]{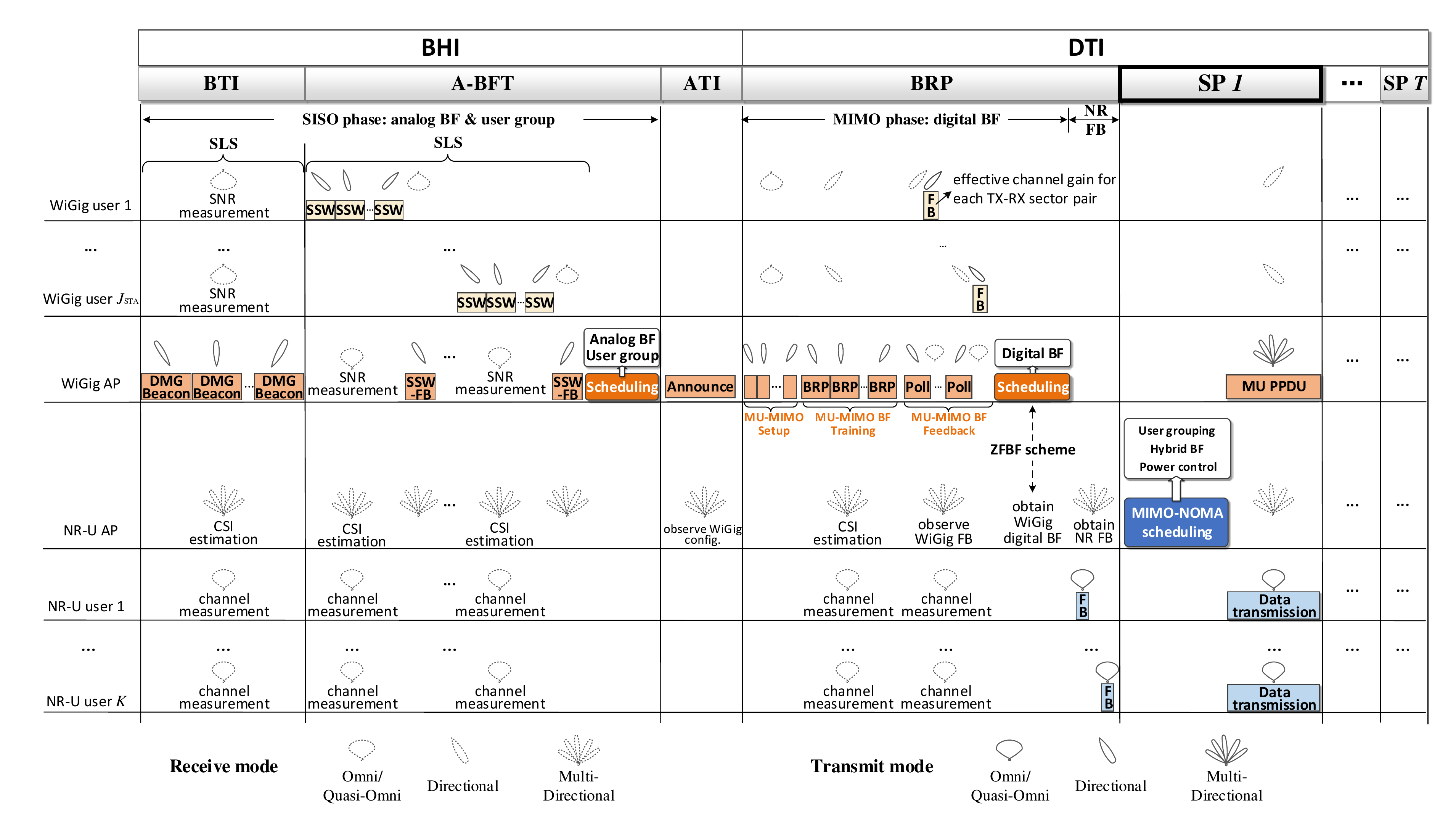}
    \caption{Flow chart of the NR-U and WiGig coexistence procedure.}
    \label{SharedAccess}
\end{figure*}

The NR-U and WiGig coexistence procedure is presented in Fig. \ref{SharedAccess}.
According to the IEEE $802.11$ay standard, the WiGig transmission period consists of multiple beacon intervals (BIs).
Each BI is divided into a beacon header interval (BHI) and a data transmission interval (DTI).
At the beginning of each BI, WiGig APs would accomplish beamforming training (BFT) \cite{Aldubaikhy2019HBF-PDVG} and WiGig user grouping.
By observing the MIMO user grouping and hybrid beamforming configuration settings of WiGig APs, NR-U AP is capable to control the inter-RAT interferences to both networks.
In what follows, we would describe the detailed procedures in the BHI and DTI.
Without loss of generality, we take downlink WiGig MU-MIMO beamforming as an example to illustrate the BFT procedure.


\subsection{Beacon Header Interval (BHI)}
Each BHI consists of a beacon transmission interval (BTI), an association beamforming training (A-BFT), and an announcement transmission interval (ATI).
WiGig APs would implement the analog BF training and the WiGig user grouping in the BHI.

During the periods of BTI and A-BFT, WiGig analog BF is first trained in the single-input single-output (SISO) phase by sequentially performing sector level sweep (SLS) at WiGig APs and users.
Specifically, in the BTI, each WiGig AP performs 
SLS by sending directional multi-gigabit (DMG) beacon frames through different sectors. Meanwhile, each WiGig user in the quasi-omnidirectional mode estimates the signal-to-noise ratio (SNR) to the WiGig AP. During the A-BFT, each WiGig user implements SLS 
by transmitting sector sweep (SSW) frames to train their sectors with the related WiGig APs, while the WiGig APs operating at the quasi-omnidirectional mode. The SSW frame contains the best sector identification (SID) of the WiGig AP with the highest SNR from the BTI phase. Thereafter, the WiGig AP utilizes the best SID to send sector sweep feedback (SSW-FB) frame, informing WiGig user about its best SID obtained from the A-BFT.
By sequentially performing SLS in the BTI and A-BFT, the best transmitting-receiving (TX-RX) analog BF beam pair for each unassociated user is obtained.

Thereafter, based on the best SID reported by WiGig users, WiGig AP can divide its associated WiGig users into different semi-orthogonal MU-MIMO user groups without requiring perfect CSI.
The detailed analog BF training and user grouping algorithm can be found in \cite{Aldubaikhy2019HBF-PDVG}.
The inter-user interferences within WiGig user group can be mitigated through the semi-orthogonal user selection, and each semi-orthogonal MU-MIMO user group would be scheduled to transmit/receive simultaneously based on MU-MIMO.

In the ATI, the WiGig AP proclaims the network management information for DTI by sending announce frame.
By observing the Extended Schedule element (ESE) entries specified in the announce frames, NR-U APs can be aware of WiGig MU-MIMO user groups, TX-RX analog beam sector configuration, and channel access scheduling during DTI.

\subsection{Data Transmission Interval (DTI)}
The DTI consists of a beam refinement period (BRP) and a series of scheduled service periods (SPs) to transmit data.

During BRP, WiGig AP obtains effective channel feedback from WiGig users, which can be used to achieve digital BF vectors that mitigate residual interference within each MU-MIMO user group.
The digital BF training is implemented through a MIMO phase in the BRP.
Similar with the IEEE $802.11$ay protocol, digital BF training is performed by MU-MIMO BF setup, MU-MIMO BF training, and MU-MIMO BF feedback. Specifically, each WiGig AP first transmits multiple MIMO BF setup frames with different sectors to declare the selected users of each MU-MIMO user group, the candidate TX sectors and their training order. Then, each WiGig AP sends BRP frames to users through the candidate TX sectors, which allows users estimating the effective channel gain utilizing the RX analog beam trained by the former SISO phase. Finally, users feedback the estimated effective channel gain under different TX sectors with its best RX sectors.
Based on the effective channel gain, each WiGig AP calculates digital BF vectors through specific strategy, e.g., SINR maximization method or zero-forcing (ZF)\cite{Aldubaikhy2019HBF-PDVG,802.11ayMUMIMOFeedback}.
In our proposed mechanism, we assume the ZF digital BF strategy is applied, and the digital BF strategy information is achieved by the NR-U AP in advance.
Therefore, the NR-U AP can also obtain the digital BF vectors of each WiGig AP by observing the effective channel feedback from WiGig users.

When WiGig devices (WiGig APs and users) transmit in the BRP in different sectors, the NR-U APs and NR-U users monitor the channel and measure the channel responses from different WiGig devices.
In detail, from the channel observation, the NR-U APs can estimate CSI by utilizing the compressed sensing based method, which recovers interference path gain and angle-of-arrival (AoA) exploiting the sparse scattering nature of mmWave channel \cite{Li2018mmWaveCSIEstimation}.
After the WiGig MU-MIMO BF training procedure, the observed channel information at NR-U users is finally fed back to NR-U AP to estimate the effective interference channel from WiGig devices to NR-U users under different WiGig antenna sectors.
Based on the channel reciprocity, effective interference channels from NR-U AP/users to WiGig AP/users can also be attained.

\section{System Model}
We consider $I$ NR-U APs co-sited with $J_{\mathrm{AP}}$ WiGig APs over the same mmWave band.
In the coverage of the NR-U networks, there are $K$ densely deployed  low-cost NR-U users, and each of them equips a single antenna.
Denote the set of NR-U APs and NR-U users as $\mathcal{I}$ and $\mathcal{K}$, respectively.
To increase antenna gain as well as reducing hardware complexity, each NR-U AP adopts the commonly used fully-connected hybrid architecture with $M_{0}$ uniform linear antenna (ULA) antennas connected to $N_{0}$ RF chains via  $M_{0}\times N_{0}$  phase shifters, $M_{0} > N_{0}$.
The analog and digital beamformers at the NR-U AP $i$ are denoted as $\mathbf{F}_i \in \mathbb{C}^{M_{0} \times N_{0}}$ and $\mathbf{D}_i \in \mathbb{C}^{N_{0} \times N_{0}}$, respectively.
Hence, the aggregate digital beamformer matrix can be denoted as $\mathbf{D} = [\mathbf{D}_1, \mathbf{D}_2, ..., \mathbf{D}_{I}]$.
By adopting NOMA technique, the highly angular-correlated NR-U users can be clustered into one user group, which are served by the same beam.
Here, we denote $\mathcal{U} = \mathcal{N}_1^{\mathrm{C}} + \mathcal{N}_2^{\mathrm{C}} + ... + \mathcal{N}_{I}^{\mathrm{C}}$ as the aggregate beam set for all the NR-U APs, where $\mathcal{N}_{i}^{\mathrm{C}}$ denotes the beam set for NR-U AP $i$.
Moreover, we denote $x_{u,k}(t) = 1$ if NR-U user $k \in \mathcal{K}$ is grouped into beam $u$, $u \in \mathcal{U}$, at SP $t$, and $x_{u,k}(t) = 0$ otherwise.
On the other hand, we denote there is a set  $\mathcal{J}=\mathcal{J}_{\mathrm{A}}\cup\mathcal{J}_{\mathrm{U}}$ of WiGig devices, which includes $\mathcal{J}_{\mathrm{A}}$ WiGig APs and $\mathcal{J}_{\mathrm{U}}$ WiGig users.
Each WiGig AP supports short transmission range, directional and extremely high-speed services.
Each WiGig AP $j\in\mathcal{J}_{\mathrm{A}}$ under hybrid beamforming mode equips $M_{j}^{\mathrm{A}}$ steerable antennas and $N_{j}$ RF chains, which can support up to $N_{j}$ spatial streams.
Moreover, each WiGig user $j\in\mathcal{J}_{\mathrm{U}}$ implements analog beamforming mode with $M_{j}^{\mathrm{U}}$ steerable antennas and a single RF chain.
The coverage of each WiGig AP/user is divided into multiple TX/RX sectors, indexed by different SIDs. Each sector corresponds to a specific analog beam, i.e., analog antenna weight vector (AWV), which comes from the pre-defined analog BF codebook.
Denote the digital beamformer for WiGig AP $j$ as $\mathbf{V}_{j}^{\mathrm{BB}} \in \mathbb{C}^{N_{j} \times N_{j}}$,
the analog beamformer for WiGig AP (user) $j$ as $\mathbf{V}_{j}^{\mathrm{RF,A}} \in \mathbb{C}^{M_{j}^{\mathrm{A}} \times N_{j}^{\mathrm{A}}}$ ($\mathbf{v}_{j}^{\mathrm{RF,U}} \in \mathbb{C}^{M_{j}^{\mathrm{U}} \times 1}$),
and the adopted hybrid beamformer of WiGig AP (user) as $\mathbf{V}_{j}^{\mathrm{A}} \in \mathbb{C}^{M_{j}^{\mathrm{A}} \times N_{j}^{\mathrm{A}}}$ ($\mathbf{V}_{j}^{\mathrm{U}} \in \mathbb{C}^{M_{j}^{\mathrm{U}} \times 1}$).
Hence, we have $\mathbf{V}_{j}^{\mathrm{A}} = \mathbf{V}_{j}^{\mathrm{RF,A}}\mathbf{V}_{j}^{\mathrm{BB}}$ for each WiGig AP $j\in\mathcal{J}^{\mathrm{A}}$
and $\mathbf{v}_{j}^{\mathrm{U}} = \mathbf{v}_{j}^{\mathrm{RF,U}}$ for each WiGig user $j\in\mathcal{J}^{\mathrm{U}}$.

Assume each DTI has $T$ SPs, indexed by $\mathcal{T} = \{1,2,...,T\}$, and the duration of each SP $t$ is denoted as $\tau(t)$.
Denote $\mathcal{J}_{\mathrm{T}}(t)$ and $\mathcal{J}_{\mathrm{R}}(t)$ as the set of WiGig TX and RX devices scheduled in SP $t$.
For the sake of expression, we take downlink WiGig MU-MIMO transmissions as example to illustrate the coexistence network, where the WiGig TX and RX devices are chosen from WiGig APs and users, respectively.
It's worth mentioning that the model formulation can be directly extended to incorporate both uplink and downlink WiGig MU-MIMO transmissions.


\subsection{Channel Model}
Define $\mathbf{h}_{i,k} \in \mathbb{C}^{M_{0} \times 1}$ as the spatial domain channel vector from NR-U user $k$ to the NR-U AP $i$,
and $\mathbf{G}_{i,j}^{\mathrm{C}} \in \mathbb{C}^{M_{0}\times M_{j}^{\mathrm{A}}}$ as the interference channel from WiGig AP $j$ to the NR-U AP $i$. Moreover, we denote $\mathbf{g}_{j,k}^{\mathrm{W}} \in \mathbb{C}^{M_{j}^{\mathrm{U}} \times 1}$ as the interference channel from NR-U user $k$ to WiGig user $j$.

According to the widely used Saleh-Valenzuela channel model \cite{Dai2019HybridPrecodingMIMONOMA}, we have
\begin{equation}\label{CellularChannel}
\small{
\mathbf{h}_{i,k} = \sum_{l=0}^{L_{i,k}} \beta_{i,k,(l)} \mathbf{a}_{0}\left( \psi_{i,k,(l)}\right), \forall k \in \mathcal{K},}
\end{equation}
where $L_{i,k}$ denotes the total number of multipath components (MPCs) from NR-U user $k$ to the NR-U AP $i$, and $\beta_{i,k,(l)}$ denotes the complex coefficient for the $l$-th MPC.
$\psi_{i,k,(l)} = \frac{d}{\lambda} \sin \varphi_{i,k,(l)}$ denotes the spatial direction,
and $\varphi_{i,k,(l)}$ is the corresponding AoA at the NR-U AP.
$\mathbf{a}_{i}(\psi)$, $i \in \{0\} \cup \mathcal{J}$, is $M_{i} \times 1$ array steering of ULA evaluated at the corresponding angle of arrival/departure (AoA/AoD),
given by
\begin{equation}\label{arrayResponse}
\small{
\mathbf{a}_{i}(\psi) = \frac{1}{M_{i}}\left[1, e^{j\pi \sin(\psi)}, ..., e^{j (N-1)\pi \sin(\psi)}\right]^T,}
\end{equation}
\noindent where $i=0$ refers to NR-U AP and $i\in\mathcal{J}$ refers to the $i$-th WiGig devices.

Similarly, the interference channel matrix $\mathbf{G}_{i,j}^{\mathrm{C}}$ from WiGig AP $j\in\mathcal{J}_{\mathrm{A}}$ to the NR-U AP $i\in \mathcal{I}$, and the interference channel vector $\mathbf{g}_{j,k}^{\mathrm{W}}$ from NR-U user $k\in\mathcal{K}$ to the WiGig device $j\in\mathcal{J}_{\mathrm{U}}$ can be respectively formulated as
\begin{equation}\label{CellularInterferenceChannel}
\small{
\mathbf{G}_{i,j}^{\mathrm{C}} \!=\!
\sum_{l=0}^{L_{i,j}^{\mathrm{C}}} \alpha_{i,j,(l)}^{\mathrm{C}} \mathbf{a}_{0}\left( \vartheta_{i,j,(l)}\right) \left(\mathbf{a}_{j}^{\mathrm{A}}\!\left( \phi_{j,i,(l)}^{\mathrm{C}}\right)\right)^{H}\!\!,
\forall i \in \mathcal{I},
\! j \in \mathcal{J}_{\mathrm{A}},}
\end{equation}
\begin{equation}\label{WiGigChannel}
\small{\mathbf{g}_{j,k}^{\mathrm{W}} \!=\! \sum_{l=0}^{L_{j,k}^{\mathrm{W}}} \alpha_{j,k,(l)}^{\mathrm{W}} \mathbf{a}_{j}^{\mathrm{U}}\left( \phi_{j,k,(l)}^{\mathrm{W}}\right), \forall k \in \mathcal{K},
 j \in \mathcal{J}_{\mathrm{U}},}
\end{equation}
\noindent where $L_{i,j}^{\mathrm{C}}$ indexes the number of MPCs from WiGig AP $j\in\mathcal{J}_{\mathrm{A}}$ to the NR-U AP $i\in \mathcal{I}$, $L_{j,k}^{\mathrm{W}}$ denotes the number of MPCs from NR-U user $k\in\mathcal{K}$ to the WiGig device $j\in\mathcal{J}_{\mathrm{U}}$, and $\alpha_{i,j,(l)}^{\mathrm{C}}$ and $\alpha_{j,k,(l)}^{\mathrm{W}}$ represent the corresponding complex gains in the $l$-th MPC. $\mathbf{a}_{0}\left( \vartheta_{i,j,(l)}^{\mathrm{C}}\right) \in \mathbb{C}^{M_{0}\times1}$  is the array response for NR-U AP $i$ evaluated at the corresponding AoA,
and $\mathbf{a}_{j}^{A}(\phi_{(l)}) \in \mathbb{C}^{M_{j}^{\mathrm{A}}\times1}$ and $\mathbf{a}_{j}^{U}(\phi_{(l)}) \in \mathbb{C}^{M_{j}^{\mathrm{U}}\times1}$ are the antenna array steering vector of WiGig AP and WiGig user at the corresponding AoA and AoD, respectively.

\subsection{WiGig Interference Description}
To improve the mmWave connections and ensure coexistence, we should mitigate the interference incurred by the uplink mmWave NR-U transmissions to the WiGig receivers.
Here, the highly-directional WiGig transmission is protected by appropriate user scheduling and power control for NR-U user.
The aggregate estimated interference power $I_{j}^{\mathrm{W}}(t)$ caused by  NR-U users to each WiGig user $j$ should be less than a predefined maximum inter-RAT interference threshold $I_{\max}^{\mathrm{W}}$ as
\begin{equation}\label{WiGigInterferenceMitigation}
\small{
I_{j}^{\mathrm{W}}(t) = \sum_{k\in\mathcal{K}}\sum_{u\in\mathcal{U}} x_{u,k} p_{k}^{\mathrm{C}} \|\left(\mathbf{v}_{j}^{\mathrm{U}}\right)^H  \mathbf{g}_{j,k}^{\mathrm{W}}\|^2 \leq I_{\max}^{\mathrm{W}},~
\forall j \in \mathcal{J}_{\mathrm{R}}(t),}
\end{equation}
\noindent where $\mathcal{J}_{\mathrm{R}}(t)$ represents the set of receiving WiGig users at SP $t$.

MmWave NR-U users can coexist with the incumbent WiGig devices when \eqref{WiGigInterferenceMitigation} is satisfied. The data transmission of NR-U AP can be synchronized with WiGig by observing the transmission occupation duration $\tau(t)$ of each SP $t$.


\subsection{Data Transmission}
The NR-U AP maintains a traffic buffer queue $Q_{k}(t)$ evolving for NR-U user $k \in \mathcal{K}$ at SP $t$ as
\begin{equation}\label{Queue}
\small{
Q_{k}(t+1) = [Q_{k}(t) - R_{k} (t) \tau(t)]^+ + A_k(t), \forall k \in \mathcal{K},  \forall t \in \mathcal{T},}
\end{equation}
\noindent where $Q_{k}(0)=0$,
$R_{k}(t)$ denotes the data rate of NR-U user $k$ during SP $t$, $\tau(t)$ denotes the transmission duration of SP $t$, and $A_k(t)$ denotes the traffic arrival volume.

Intuitively, the queue lengths influence NR-U users' latency, which can be compensated by improving their transmission data rates\footnote{Since the durations of BFT and WiGig user grouping would not influence the optimal result, we have ignored them in the calculation of latency.}.
Therefore, we novelly propose to determine SIC decoding order based on the queue lengths of $Q_k(t)$ and $Z_k^{\mathrm{C}}(t), \forall k$, where $Z_{k}^{\mathrm{C}}(t)$ is the virtual queue that would be introduced later.
Specifically, NR-U users with larger queue backlogs would be decoded later to suppress intra-beam interferences from users with shorter queue backlogs, and thus increase data rates of users with larger queues.
At each frame $t$, the packets of NR-U users are re-indexed and decoded in an ascending order of queue backlogs as
\begin{equation}
\small{
Q_{1}(t) + Z_{1}^{\mathrm{C}}(t)
\leq Q_{2}(t) + Z_{2}^{\mathrm{C}}(t)
\leq \cdots
 \leq Q_{K}(t) + Z_{K}^{\mathrm{C}}(t).}
\end{equation}

For each RF chain $u \in \mathcal{U}$, define $I(u) = i$, if $u \in \mathcal{N}_{i}^{\mathrm{C}}$, i.e., beam $u$ corresponds to NR-U AP $i \in \mathcal{I}$.
With the queue-aware decoding order, the uplink received signal after SIC for NR-U user $k$ associated with the NR-U AP $i$ on RF chain $u \in \mathcal{N}_{i}^{\mathrm{C}}$ can be written as
\begin{equation}\label{signal}
\small{
\begin{split}
 y_{u,k} (t) &=
 \underbrace{x_{u,k}  \mathbf{d}_{u}^H \mathbf{F}_{I(u)}^{H} \mathbf{{h}}_{I(u),k} \sqrt{p_{k}^{\mathrm{C}}} s_{k}^{\mathrm{C}}}_{\text{desired signal}}
 \\ & + \underbrace{ \sum_{k'=k+1}^{K} x_{u,k'} \mathbf{d}_{u}^{H} \mathbf{F}_{I(u)}^{H} \mathbf{{h}}_{I(u),k'} \sqrt{p_{k'}^{\mathrm{C}}} s_{k'}^{\mathrm{C}}}_{\text{intra-beam interference}}
 \\ & + \underbrace{ \sum_{\substack{u'\ne u}} \sum_{\substack{k'\ne k}} x_{u',k'} \mathbf{d}_{u}^H \mathbf{F}_{I(u)}^{H} \mathbf{{h}}_{I(u),k'} \sqrt{p_{k'}^{\mathrm{C}}}  s_{k'}^{\mathrm{C}}}_{\text{inter-beam/AP interference}}
 \\ &+ \underbrace{ \sum_{j\in\mathcal{J}_\mathrm{T}(t)} \mathbf{d}_{u}^H \mathbf{F}_{I(u)}^{H} \mathbf{{G}}_{I(u),j}^{\mathrm{C}} \mathbf{V}_{j}^{\mathrm{A}} \sqrt{p_{j}^{\mathrm{W}}} \mathbf{s}_{j}^{\mathrm{W}}}_{\text{inter-RAT interference}}
+ n_{0},
\end{split}}
\end{equation}
where $s_{k}^{\mathrm{C}}$ denotes the data symbol of NR-U user $k$, $p_{k}^{\mathrm{C}}$ represents the power control for NR-U user $k$,
and $\mathbf{d}_{u}$ means the $u$-th column of digital beamformer $\mathbf{D}=[\mathbf{D}_1,\mathbf{D}_2,...,\mathbf{D}_I]$.
$\mathbf{s}_{j}^{\mathrm{W}}$ is the data symbol of WiGig AP $j$, and $\mathbf{V}_{j}^{\mathrm{A}}$ is the corresponding hybrid beamformer.
$n_{0} \sim \mathcal{CN}(0,\sigma^2)$ is the noise and $\sigma^2$ is the noise power spectral density.

Therefore, the data rate of NR-U user $k$ transmitted in beam $u \in \mathcal{U}$ can be formulated as
\begin{equation}\label{mmW_rate}
\small{
\begin{split}
R_{u,k}(t)=& \log_2 \left(1 + SINR_{u,k}(t) \right) \\
=& \log_2 \bigg(1 + x_{u,k} |\mathbf{d}_{u}^H \mathbf{F}_{I(u)}^H \widetilde{\mathbf{h}}_{u,k}|^2 p_{k}^{\mathrm{C}}\times
\\&\bigg(\sum_{(u',k')\in\mathcal{S}_{u,k}} x_{u',k'} |\mathbf{d}_{u}^H \mathbf{F}_{I(u)}^H \widetilde{\mathbf{h}}_{u,k'}|^2 p_{k'}^{\mathrm{C}}
\\&+ \sum_{j\in\mathcal{J}_{\mathrm{T}}} \|\mathbf{d}_{u}^H \mathbf{F}_{I(u)}^H \widetilde{\mathbf{G}}_{u,j}^{\mathrm{C}}\|^2 p_{j}^{\mathrm{W}} +{\sigma}^2\bigg)^{-1} \bigg),
\end{split}}
\end{equation}
where $SINR_{u,k}(t)$ is the signal-to-interference-and-noise ratio (SINR) of NR-U user $k$ in user group $u$, $\mathcal{S}_{u,k} = \{ (u',k') | u'=u, k' > k \} \cup \{ (u',k') | u' \ne u,  k' \ne k \}$ is the set of spatial streams causing intra-beam and inter-beam interference to NR-U user $k$ in user group $u$ after SIC,
$\widetilde{\mathbf{h}}_{u,k} = \mathbf{h}_{I(u),k}$, and $\widetilde{\mathbf{G}}_{u,j}^{\mathrm{C}} = \mathbf{G}_{I(u),j}^{\mathrm{C}} \mathbf{V}_{j}^{\mathrm{A}}$.
The achievable data rate of NR-U user $k$ during SP $t$ can be expressed as $R_{k}(t) = \sum_{u\in\mathcal{U}} R_{u,k}(t)$.

\section{Problem Formulation}
Our goal in this work is to design a joint user grouping, beam coordination and power control strategy to realize massive connectivity and guarantee latency for the NR-U network as well as ensuring performance of WiGig networks. The optimization problem can be formulated as
\setcounter{equation}{9}
\begin{equation*}
\small{
\mathcal{P}1:~~ \max_{\boldsymbol{\omega}(t)} \lim _{T\to\infty}  \frac{1} {T} \sum_{t=0}^{T-1}  \mathbb{E} } \left[R(t)\right] \tag{\ref{P1}}
\end{equation*}
\begin{subequations}\label{P1}
subject to \eqref{WiGigInterferenceMitigation},
\begin{equation}\label{AvgWiGigInterference}
\small{\lim _{T\to\infty} \frac{1}{T} \sum_{t=0}^{T-1} \mathbb{E}[I^{\mathrm{W}}(t)] \leq \bar{I}^{\mathrm{W}}, ~\forall k \in \mathcal{K},}
\end{equation}
\begin{equation}\label{Delay}
\small{\lim _{T\to\infty} \frac{1}{T} \sum_{t=0}^{T-1} \mathbb{E}[Q_{k}(t)] \leq \bar{Q}_{k}, ~\forall k \in \mathcal{K},}
\end{equation}
\begin{equation}\label{MinRate}
\small{x_{u,k}R_{\min} \leq R_{u,k}(t) \leq \frac{Q_{k}(t)}{\tau(t)},~ \forall u \in \mathcal{U}, k \in \mathcal{K}, t \in \mathcal{T},}
\end{equation}
\begin{equation}\label{Power}
\small{0 \leq p_{k}^{\mathrm{C}}(t) \leq \sum_{u\in\mathcal{U}} x_{u,k}(t) P^{\max},~ \forall k \in \mathcal{K}, t \in \mathcal{T},}
\end{equation}
\begin{equation}\label{AnalogBF}
\small{|\mathbf{F}_{i}(m,n)(t)| \!=\! 1, ~\forall  m \in \{1,2,..,M_{0}\}, n \in \mathcal{N}_{i}^{\mathrm{C}}, i\in\mathcal{I}, t \in \mathcal{T},}
\end{equation}
\begin{equation}\label{ConnPerDevice}
\small{\sum_{u\in\mathcal{U}} x_{u,k}(t) = 1,~ \forall k \in \mathcal{K}, t \in \mathcal{T},}
\end{equation}
\begin{equation}\label{Bin}
\small{x_{u,k}(t) \in \{0,1\},~ \forall u \in \mathcal{U}, k \in \mathcal{K}, t \in \mathcal{T},}
\end{equation}
\end{subequations}
where $\boldsymbol{\omega}(t) = [\mathbf{X}(t),\mathbf{F}(t),\mathbf{D}(t),\mathbf{p}^{\mathrm{C}}(t)]$. $R(t) = \sum_{u\in\mathcal{U}} \sum_{k\in\mathcal{K}} R_{u,k}(t)$ denotes the spectrum efficiency of NR-U network in SP $t$.
$I^{\mathrm{W}}(t) = \sum_{j\in\mathcal{J}_{\mathrm{R}(t)}} I_{j}^{\mathrm{W}}(t)$ denotes the inter-RAT interference suffered by all the WiGig users.
Constraint \eqref{WiGigInterferenceMitigation} limits the inter-RAT interference suffered by each WiGig users,
and constraint \eqref{AvgWiGigInterference} restricts the long-term average inter-RAT interference to WiGig.
Constraint \eqref{Delay} ensures the time-averaged latency of NR-U users, which is proportional to the average queue length according to the Little's theorem \cite{Bertsekas1992LittleLaw}.
Constraint \eqref{AnalogBF} signifies the unit-modulus constraint of analog beamforming.
Constraint \eqref{MinRate} guarantees the minimum data rate requirement of NR-U users, and ensures the data rate less than the queue backlog size.
Moreover, constraint \eqref{Power} denotes the transmit power limitation and constraint \eqref{ConnPerDevice} guarantees that each NR-U user occupies one RF chain.

To tackle the infinite horizon problem $\mathcal{P}1$,
we utilize Lyapunov optimization to transform $\mathcal{P}1$ into a tractable single frame problem.
We define virtual queues $Z^{\mathrm{W}}(t)$ and $Z_{k}^{\mathrm{C}}(t)$ for constraints \eqref{AvgWiGigInterference} and  \eqref{Delay} as
\begin{equation}\label{ZWiGig}
\small{Z^{\mathrm{W}}(t+1) = [ Z^{\mathrm{W}}(t)  -   \bar{I}^{\mathrm{W}} + I^{\mathrm{W}}(t)]^+, ~\forall t \in \mathcal{T},}
\end{equation}
\begin{equation}\label{ZNR}
\small{Z_{k}^{\mathrm{C}}(t+1) = [ Z_{k}^{\mathrm{C}}(t) -  \bar{Q}_{k} + Q_{k}(t+1) ]^+, ~\forall k \in \mathcal{K}, t \in \mathcal{T},}
\end{equation}
\noindent where $Z^{\mathrm{W}}(0) = 0$ and $Z_{k}^{\mathrm{C}}(0) = 0, \forall k\in\mathcal{K}$.

The Lyapunov function can be formulated as a sum of virtual queue length squares, given by
\begin{equation}\label{LyaounovFun}
\small{
L(\mathbf{Z}(t)) =
\frac{1}{2} \left[\sum_{k\in\mathcal{K}} (Z_{k}^{\mathrm{C}}(t))^2 + (Z^{\mathrm{W}}(t))^2 \right],}
\end{equation}
where $\mathbf{Z}(t) = [Z_{1}^{\mathrm{C}}(t), Z_{2}^{\mathrm{C}}(t), ..., Z_{K}^{\mathrm{C}}(t), Z^{\mathrm{W}}(t)]$.

Moreover, the Lyapunov drift function is expressed as
\begin{equation}\label{LyapunovDrift}
\small{\Delta \mathbf{Z}(t) = \mathbb{E} \left[ L(\mathbf{Z}(t+1)) - L(\mathbf{Z}(t)) \mid \mathbf{Z}(t) \right].}
\end{equation}

Denote $\mathbf{\Theta}(t) = [\mathbf{Q}(t),\mathbf{Z}(t)]$, where $\mathbf{Q}(t) = \left[Q_{1}(t),Q_{2}(t), ..., Q_K(t)\right]$.
Then, the drift-plus-penalty is defined to minimize a bound on the drift-plus-penalty expression by selecting control actions, which can be formulated as
\begin{equation}\label{DriftPenalty}
\small{ \Delta_{v}(t) =
 V_{0}\sum_{k\in\mathcal{K}}\Delta Z_{k}^{\mathrm{C}}(t) +
 V_{1}\Delta Z^{\mathrm{W}}(t) -
 \mathbb{E} [R(t)|\mathbf{\Theta}(t)],}
\end{equation}
\noindent where $V_0$ and $V_1$ are non-negative parameters controlling the tradeoff between WiGig performance degradation, NR-U traffic delay and spectral efficiency.

We then derive an upper bound for $\Delta_{v}(t)$ in the following lemma, which would play a critical role throughout the analysis of $\mathcal{P}1$.

\begin{lemma}\label{Lemma_DriftPenalty}
The upper bound for the drift-plus-penalty \eqref{DriftPenalty} can be written as
\begin{equation}\label{DriftPenaltyUpperBound}
\small{
\begin{split}
& \Delta_{v}(t)
\leq B + \mathbb{E} \bigg[ V_{1}I^{\mathrm{W}}(t)\left(Z^{\mathrm{W}}(t)-\bar{I}^{\mathrm{W}}\right) -
\\
& V_{0}\sum_{k\in\mathcal{K}} \left(Q_{k}(t)+Z_{k}^{\mathrm{C}}(t)\right) R_{k}(t)\tau(t)
-R(t) \mid \mathbf{\Theta}(t) \bigg],
\end{split}}
\end{equation}
\noindent where $B=B_1+B_2$.
$B_1 = V_0\frac{\sum_{k\in\mathcal{K}} \left[(\bar{Q}_{k})^2+  (A_{k}^{\max})^2 \right]}{2}
+ V_{1}\frac{(\bar{I}^{\mathrm{W}}) ^2  + (|\mathcal{J}_{\mathrm{R}}| \gamma_{\mathrm{W}})^2}{2} $ is constant,
and $B_2 = V_{0}\sum_{k\in\mathcal{K}} [ Q_{k}^2(t) + Z_{k}^{\mathrm{C}}(t)(Q_{k}(t) - \bar{Q}_{k})]$ is fixed at each frame.
\end{lemma}
{\it{Proof}:} See Appendix A.

Note that the upper bound of $\Delta_{v}(t)$ coincides with the objective function in $\mathcal{P}1$. Therefore, we transfer $\mathcal{P}1$ into the following auxiliary problem $\mathcal{P}2$ to facilitate the performance analysis.
\begin{equation}\label{P2}
\small{
\begin{split}
\mathcal{P}2:
\max_{\boldsymbol{\omega}(t)}~ &
\bigg[ R(t)
+ V_{0}\sum_{k\in\mathcal{K}} (Q_{k}(t) + Z_{k}^{\mathrm{C}}(t)) R_{k}(t)\tau(t)
 \\&
- V_{1}Z^{\mathrm{W}}(t) I^{\mathrm{W}}(t)  \bigg],
\end{split}}
\end{equation}
subject to \eqref{WiGigInterferenceMitigation}, \eqref{MinRate} - \eqref{Bin}.

At the beginning of each SP $t$, NR-U network observes system environment and solves $\mathcal{P}_2$ to dynamically and jointly schedule user grouping, coordinated hybrid beamforming and power control.
In the following section, we would ignore the SP index $t$ to present a optimization solution for $\mathcal{P}_2$ for conciseness.

\section{Joint User Selection, Beam Coordination and Power Control}
Intuitively, $\mathcal{P}2$ is a nonconvex MINLP with coupling constraints, which is NP-hard.
To jointly solve this user grouping, coordinated hybrid beamforming and power control problem, a \emph{PDD-CCCP} method is proposed to find a desirable solution.
First, we equivalently transfer the nonconvex SINR expression and the binary user grouping constraints by introducing auxiliary variables.
Thereafter, the transformed function is further recast into an AL problem by penalizing and dualizing the coupling equality constraints.
Since the AL problem can be deemed as a difference of convex (D.C.) programming,
we jointly solve it by incorporating CCCP \cite{Yuille2003CCCP} into the strong-convergence guaranteed PDD framework \cite{Shi2017Penalty}.
Based on the PDD framework, the proposed method leverages a dual loop structure.
In the inner loop, the AL problem is solved by conjunctively utilizing CCCP and inexact BCU algorithms.
Specifically, the nonconvex part of the AL problem is approximately transferred with CCCP.
Thereafter, the approximated AL problem is decomposed into several nested convex subproblems by inexact BCU method, where
each subproblem can be sequentially optimized to an inexact solution at each iteration.
In the outer loop, penalty parameters and dual variables are updated.

\subsection{Problem Transformation}
The SINR expression and the unit-modulus constraints \eqref{AnalogBF} result in the  nonconvexity of $\mathcal{P}2$, which makes the problem hard to be solved with conventional schemes.
Hence, we introduce a fully digital beamformer $\mathbf{W} = [\mathbf{W}_{1},\mathbf{W}_{2},...,\mathbf{W}_{I}]\in\mathbb{C}^{M_0 \times U}$, $U=N_{0}I$, and let auxiliary variable vectors $\boldsymbol{\mu}_{u,k} = [\mu_{u,1,k}, \mu_{u,2,k}, ..., \mu_{u,U,k}] \in \mathbb{R}^{1\times U}$ and $\xi_{u,k}$ denote the effective data and interference gains from NR-U users and WiGig TX devices to the NR-U APs received by each RF chain $u$, which are defined as
\begin{equation}\label{auxiliaryBF}
\small{\mathbf{W}_{i} = \mathbf{F}_{i}\mathbf{D}_{i}, \forall i \in \mathcal{I},}
\end{equation}
\begin{equation}\label{auxiliaryMu}
\small{\boldsymbol{\mu}_{u,k} = \mathbf{x}_{k}^{T} |\mathbf{w}_{u}^{H} \mathbf{\widetilde{h}}_{u,k}|^2 p_{k}^{\mathrm{C}},
~ \forall u \in \mathcal{U}, k \in \mathcal{K},}
\end{equation}
\begin{equation}\label{auxiliaryXi}
\small{\xi_{u,j} = \|\mathbf{w}_{u}^{H} {\mathbf{\widetilde{G}}}_{u,j}^{\mathrm{C}}\|^2 p_{j}^{\mathrm{W}},
~ \forall u \in \mathcal{U}, j \in \mathcal{J}_{\mathrm{T}}.}
\end{equation}

We further denote $\gamma_{u,k}$ as the SINR of NR-U user $k$ received by RF chain $u$, then the SINR expression in \eqref{mmW_rate} can be rewritten as
\begin{equation}\label{auxiliarySINR}
\small{
\gamma_{u,k} \left(
\sum_{(u',k') \in \mathcal{S}_{u,k}}\mu_{u,u',k'}
+ \sum_{j\in\mathcal{J}_{\mathrm{T}}} \xi_{u,j}
+ {\sigma}^2 \right)
 = {\mu}_{u,u,k}.}
\end{equation}

To deal with the binary user grouping parameters, we introduce auxiliary variables $\mathbf{\widetilde{X}} = [\mathbf{\widetilde{x}}_1,...,\mathbf{\widetilde{x}}_K] \in \mathbb{C}^{U \times K}$.
Hence, the binary user grouping constraints can be equivalently replaced by the following constraints:
\begin{equation}\label{EquvX}
\small{\mathbf{x}_{k} - \mathbf{\widetilde{x}}_{k} = \mathbf{0},  ~ \forall k \in \mathcal{K},}
\end{equation}
\begin{equation}\label{EquvX2}
\small{x_{u,k}(1-\widetilde{x}_{u,k}) = 0, ~ \forall u \in \mathcal{U},k \in \mathcal{K},}
\end{equation}
\begin{equation}\label{ContinuousVar}
\small{\mathbf{0} \leq \mathbf{\widetilde{x}}_{k} \leq \mathbf{1}, ~\forall k \in \mathcal{K}.}
\end{equation}

In this way, $\mathcal{P}2$ can be rearranged as
\begin{equation}
\small{
\begin{split}
\mathcal{P}3:
\max_{\boldsymbol{\omega},\boldsymbol{\widehat{\omega}}}&
\sum_{k\in\mathcal{K}} \sum_{u\in\mathcal{U}}
 \left[ V_0(Q_{k}+Z_{k}^{\mathrm{C}})\tau+1 \right] \log_2(1+\gamma_{u,k}) \\&
- V_1 Z^{\mathrm{W}} \sum_{k\in\mathcal{K}} \sum_{j \in \mathcal{J}_{\mathrm{R}}} \|\mathbf{V}_{j}^H \mathbf{g}_{j,k}^{\mathrm{W}}\|^2 p_{k}^{\mathrm{C}},
\end{split}}
\tag{\ref{P4}}
\end{equation}
subject to \eqref{auxiliaryBF} - \eqref{EquvX2},
\begin{subequations}\label{P4}
\begin{equation}\label{WiGigInterferenceMitigation_P3}
\small{ \sum_{k\in\mathcal{K}} \|\mathbf{V}_{j}^H \mathbf{g}_{j,k}^{\mathrm{W}}\|^2 p_{k}^{\mathrm{C}} \leq \gamma_{\mathrm{W}},
\forall j \in \mathcal{J}_{\mathrm{R}},}
\end{equation}
\begin{equation}\label{MinRate_P3}
\small{\log_2 (1+\gamma_{u,k}) \geq x_{u,k}R_{\min}, ~ \gamma_{u,k} \leq \gamma_{k}^{\max},  \forall u \in \mathcal{U}, \forall k \in \mathcal{K},}
\end{equation}
\begin{equation}\label{Power_P3}
\small{0 \leq p_{k}^{\mathrm{C}} \leq \sum_{u\in\mathcal{U}} x_{u,k}(t) P^{\max},~ \forall k \in \mathcal{K},}
\end{equation}
\begin{equation}\label{AnalogBF_P3}
\small{|\mathbf{F}_{i}(m,n)| = 1, ~ \forall m \in \{1,2,...,M_{0}\}, n \in \mathcal{N}_{i}^{\mathrm{C}}, i \in \mathcal{I},}
\end{equation}
\begin{equation}\label{ConnPerDevice_P3}
\small{\sum_{u\in\mathcal{U}} x_{u,k} = 1,~ \forall k \in \mathcal{K},}
\end{equation}
\begin{equation}\label{ContinuousVar_P3}
\small{\mathbf{0} \leq \mathbf{\widetilde{x}}_{k} \leq \mathbf{1}, ~\forall k \in \mathcal{K},}
\end{equation}
\end{subequations}
where $\boldsymbol{\omega} = [\mathbf{X},\mathbf{D},\mathbf{F},\mathbf{p}^{\mathrm{C}}]$,
$ \boldsymbol{\widehat{\omega}} = [\mathbf{\widetilde{X}},\mathbf{W},\boldsymbol{\gamma},\boldsymbol{\mu},\boldsymbol{\xi}]$,
and $\gamma_{k}^{\max} = 2^{\frac{Q_{k}}{\tau}} - 1$.
Note that in \eqref{P4} and \eqref{WiGigInterferenceMitigation_P3}, we have utilized the fact that $\sum_{u\in\mathcal{U}} x_{u,k} p_{k}^{\mathrm{C}} = p_{k}^{\mathrm{C}}$, which is implied by \eqref{ConnPerDevice_P3}.

Resorting to the PDD framework that realizes the integration of ADMM and penalty method, $\mathcal{P}3$ can be further formulated as the following AL problem with decoupled constraints:
\begin{equation}\label{AugLagrangian}
\small{
\max_{\boldsymbol{\omega},\boldsymbol{\widehat{\omega}}} \mathcal{L}(\boldsymbol{\omega},\boldsymbol{\widehat{\omega}})
= f(\boldsymbol{\omega},\boldsymbol{\widehat{\omega}})  - \sum_{i=1}^{4} \mathcal{L}_{i}(\boldsymbol{\omega},\boldsymbol{\widehat{\omega}}),}
\end{equation}
\noindent subject to \eqref{WiGigInterferenceMitigation_P3} - \eqref{AnalogBF_P3}, and \eqref{ContinuousVar_P3}.

Here, $f(\boldsymbol{\omega},\boldsymbol{\widehat{\omega}})
 = \sum_{k\in\mathcal{K}} \sum_{u\in\mathcal{U}}
 \big[ V_0(Q_{k}+Z_{k}^{\mathrm{C}})\tau +1\big]\log_2\big(1+\gamma_{u,k}\big)  - V_1 Z^{\mathrm{W}} \sum_{k\in\mathcal{K}} \sum_{j \in \mathcal{J}_{\mathrm{R}}} \|\mathbf{V}_{j}^H \mathbf{g}_{j,k}^{\mathrm{W}}\|^2 p_{k}^{\mathrm{C}}$,
and
 $\mathcal{L}_{i}(\boldsymbol{\omega},\boldsymbol{\widehat{\omega}}), i = \{1,2,3,4\}$, are the additional penalty terms corresponding to equality constraints \eqref{ConnPerDevice_P3} and \eqref{auxiliaryBF} - \eqref{EquvX2}, given by
\begin{equation}\label{AL1}
\small{
\begin{split}
\mathcal{L}_{1}(\boldsymbol{\omega},\boldsymbol{\widehat{\omega}}) =&
\frac{1}{2\rho}\sum_{k\in\mathcal{K}}\left(\sum_{u\in\mathcal{U}}\mathbf{c}_{u}^{T}\mathbf{x}_{k}-1+\rho\lambda_{k}^{\mathrm{U}}\right)^{2}
\\&+ \frac{1}{2\rho}\sum_{k\in\mathcal{K}}\left\| \mathbf{x}_{k}-\mathbf{\widetilde{x}}_{k}+\rho\boldsymbol{\lambda}_{k}^{\mathrm{X}}\right\| ^{2}
\\& + \frac{1}{2\rho}\sum_{k\in\mathcal{K}}\sum_{u\in\mathcal{U}}\left[\mathbf{c}_{u}^{T}\mathbf{x}_{k}(1-\widetilde{x}_{u,k})+\rho\widetilde{\lambda}_{u,k}^{\mathrm{X}}\right]^{2},
\end{split}}
\end{equation}
\begin{equation}\label{AL2}
\small{
\mathcal{L}_{2}(\boldsymbol{\omega},\boldsymbol{\widehat{\omega}})
= \frac{1}{2\rho} \sum_{u} \|\mathbf{w}_{u}-\mathbf{F}_{I(u)}\mathbf{d}_{u}+\rho\boldsymbol{\lambda}^{\mathrm{W}}_{u}\|^2,}
\end{equation}
\begin{equation}\label{AL_signal}
\small{
\begin{split}
\mathcal{L}_{3}(\boldsymbol{\omega},\boldsymbol{\widehat{\omega}})
=& \frac{1}{2\rho} \sum_{u\in\mathcal{U}}\sum_{k\in\mathcal{K}}
\left\| \boldsymbol{\mu}_{u,k} - \mathbf{x}_{k}^{T} |\mathbf{w}_{u}^H \widetilde{\mathbf{h}}_{u,k}|^2 p_{k}^{\mathrm{C}}  + \rho  \boldsymbol{\lambda}_{u,k}^{\mu} \right\| ^2
\\& + \frac{1}{2\rho} \sum_{u\in\mathcal{U}} \sum_{j\in\mathcal{J}_{\mathrm{T}}} \left( \xi_{u,j} - \|\mathbf{w}_{u}^H  \mathbf{\widetilde{G}}_{u,j}^{\mathrm{C}}\|^2 p_{j}^{\mathrm{W}} + \rho \lambda_{u,j}^{\xi} \right)^2,
\end{split}}
\end{equation}
\begin{equation}\label{AL4}
\small{
\begin{split}
\mathcal{L}_{4}(\boldsymbol{\omega},\boldsymbol{\widehat{\omega}})
=& \frac{1}{2\rho} \sum_{k\in\mathcal{K}} \sum_{u\in\mathcal{U}}
    \bigg[ \gamma_{u,k} \bigg( \sum_{(u',k') \in \mathcal{S}_{u,k}} \mu_{u,u',k'}
 \\& + \sum_{j\in\mathcal{J}_{\mathrm{T}}}\xi_{u,j}
+ {\sigma}^2 \bigg)
    - \mu_{u,u,k}
    + \rho\lambda_{u,k}^{\gamma} \bigg]^2,
\end{split}}
\end{equation}
\begin{equation}
\small{
\begin{split}\label{ApproximatedAL_signal}
\mathcal{\widehat{L}}_3 (\boldsymbol{\omega},\boldsymbol{\widehat{\omega}})
=& \frac{1}{2\rho}\sum_{k\in\mathcal{K}}\sum_{u\in\mathcal{U}}
\bigg[ \|\mathbf{x}_{k} |\mathbf{w}_{u}^{H}\widetilde{\mathbf{h}}_{u,k}|^{2}p_{k}^{\mathrm{C}}\|^{2}
- 2\widehat{\boldsymbol{\zeta}}_{u,k}^{\mathrm{C}}\mathbf{x}_{k}p_{k}^{\mathrm{C}}
 +\\& \|\rho\boldsymbol{\lambda}_{u,k}^{\mu}+\boldsymbol{\mu}_{u,k}\|^{2} \bigg]
 + \frac{1}{2\rho}\sum_{u\in\mathcal{U}}\sum_{j\in\mathcal{J}_{\mathrm{T}}} \bigg[ (\|\mathbf{w}_{u}^{H}\mathbf{\widetilde{G}}_{u,j}^{\mathrm{C}}\|^{2}p_{j}^{\mathrm{W}})^{2}
\\& -2p_{j}^{\mathrm{W}}\widehat{\zeta}_{u,j}^{\mathrm{W}}
+(\rho\lambda_{u,j}^{\xi}+\xi_{u,j})^{2} \bigg],
\end{split}}
\end{equation}
where $\rho$ is the non-negative penalty parameter, and $\boldsymbol{\lambda} = \left[\{\lambda_{u,k}^{\mathrm{U}}\}, \{\lambda_{u,k}^{\mathrm{X}}\}, \{\widetilde{\lambda}_{u,k}^{\mathrm{X}}\}, \{\boldsymbol{\lambda}_{u}^{\mathrm{W}}\}, \{\lambda_{u,u',k}^{\mu}\},\{\lambda_{u,j}^{\xi}\},\{\lambda_{u,k}^{\gamma}\}\right]$ denotes the dual variables related to equality constraints \eqref{ConnPerDevice_P3} and the coupled equality constraints \eqref{auxiliaryBF} - \eqref{EquvX2}.
Moreover, $\mathbf{c}_{u}$ in \eqref{AL1} denotes the $u$-th column of identity matrix $\mathbf{I}_{U\times U}$, and $\mathbf{c}_{u}^{T}\mathbf{x}_{k} = x_{u,k}$.

Now, our goal is to find the optimal user grouping, hybrid beam coordination and power control strategy in the AL problem \eqref{AugLagrangian}.
Since the nonconvex term \eqref{AL_signal} can be regarded as a typical D.C. functions \cite{Yuille2003CCCP} with respect to the fully digital beamformer $\mathbf{W}$,
we incorporate CCCP method into PDD framework to transform the nonconvex part of the AL problem.
At each iteration of the inner loop in PDD framework, the AL problem is approximately converted into a multi-convex programming based on CCCP under fixed $\boldsymbol{\lambda}^{\mathrm{\mu}}$ and $\boldsymbol{\lambda}^{\mathrm{\xi}}$.
Thereafter, the approximated multi-convex AL problem is decomposed into multiple blocks, sequentially and inexactly solved by inexact BCU method.
The inexact BCU method does not require an exact solution to be obtained with respect to $\mathbf{W}$ during each block update, and thus can accelerate the whole convergence process.
We introduce the following Lemma to obtain the approximated AL problem based on the CCCP definition.


\begin{lemma}\label{Lemma_ApproximatedAL}
The AL problem can be approximated by maximizing the lower bound of $\mathcal{L}(\boldsymbol{\omega},\boldsymbol{\widehat{\omega}})$, given by
\setcounter{equation}{31}
\begin{equation}\label{ApproximatedAL}
\small{
\max_{\boldsymbol{\omega},\boldsymbol{\widehat{\omega}}}
\mathcal{\widehat{L}}(\boldsymbol{\omega},\boldsymbol{\widehat{\omega}})
\!=\! \max_{\boldsymbol{\omega},\boldsymbol{\widehat{\omega}}} \bigg( f(\boldsymbol{\omega},\boldsymbol{\widehat{\omega}})
- \! \!  \! \sum_{i \in \{1,2,4\}}\!  \! \mathcal{L}_{i}(\boldsymbol{\omega},\boldsymbol{\widehat{\omega}})
- \mathcal{\widehat{L}}_{3}(\boldsymbol{\omega},\boldsymbol{\widehat{\omega}}) \bigg),}
\end{equation}
subject to \eqref{WiGigInterferenceMitigation_P3} - \eqref{ContinuousVar_P3}.

Here, $\mathcal{\widehat{L}}_3 (\boldsymbol{\omega},\boldsymbol{\widehat{\omega}})$ in \eqref{ApproximatedAL} is the approximation of $\widehat{L}_3 (\boldsymbol{\omega},\boldsymbol{\widehat{\omega}})$ defined as \eqref{ApproximatedAL_signal},  where $\widehat{\boldsymbol{\zeta}}_{u,k}^{\mathrm{C}}=[\widehat{\zeta}_{u,1,k}^{\mathrm{C}},\widehat{\zeta}_{u,2,k}^{\mathrm{C}},...,\widehat{\zeta}_{u,U,k}^{\mathrm{C}}]$.
Under fixed $\boldsymbol{\lambda}^{\mathrm{\mu}}$ and $\boldsymbol{\lambda}^{\mathrm{\xi}}$, $\{\widehat{\zeta}_{u,u',k}^{\mathrm{C}}\}$ and $\{\widehat{\zeta}_{u,j}^{\mathrm{W}}\}$ are given by
\begin{equation}\label{ApproximatedZeta}
\begin{split}
\widehat{\zeta}_{u,q} =&
(a_{u,q}+\rho[\lambda_{u,q}]^{+})\left[ 2\mathfrak{Re}\{\bar{\mathbf{w}}_{u}^{H}\mathbf{Y}_{u,q}\mathbf{Y}_{u,q}^{H}\mathbf{w}_{u}\}
\right.\\&\left.-\|\bar{\mathbf{w}}_{u}^{H}\mathbf{Y}_{u,q}\|^{2}\right]
-\rho[-\lambda_{u,q}]^{+}\|\mathbf{w}_{u}^{H}\mathbf{Y}_{u,q}\|^{2},
\end{split}
\end{equation}
\noindent where $\bar{\mathbf{w}}_{u}$ denotes the iterative point of $\mathbf{w}_{u}$ obtained from the previous BCU iteration, and $\mathfrak{Re}\{\cdot\}$ stands for the real component of complex number.
For NR-U user $k$ associated with RF chain $u' \in \mathcal{U}$,
we have $a_{u,q} = \mu_{u,u',k}$,
$\lambda_{u,q} = \lambda_{u,u',k}^{\mu}$,
and $\mathbf{Y}_{u,q} = \widetilde{\mathbf{h}}_{u,k}$.
For WiGig device $j \in \mathcal{J}_{\mathrm{T}}$,
we have $a_{u,q} = \xi_{u,j}$,
$\lambda_{u,q} = \lambda_{u,j}^{\xi}$,
and $\mathbf{Y}_{u,q} = \mathbf{\widetilde{G}}_{u,j}^{\mathrm{C}}$.
\end{lemma}
\begin{proof}
Define $\boldsymbol{\zeta}_{u,k}^{\mathrm{C}}=(\boldsymbol{\mu}_{u,k}+\rho\boldsymbol{\lambda}_{u,k}^{\mu}) |\mathbf{w}_{u}^{H} \widetilde{\mathbf{h}}_{u,k}|^2$
and ${\zeta}_{u,j}^{\mathrm{W}} = (\xi_{u,j} + \rho \lambda_{u,j}^{\xi})  \|\mathbf{w}_{u}^H \mathbf{\widetilde{G}}_{u,j}^{\mathrm{C}}\|^2$.
By linearizing $\mathbf{w}_n$ with first-order Taylor approximation around its previous iterative point, we have
$\|\mathbf{w}_{u}^{H}\mathbf{Y}_{u,q}\|^{2}\geq2\mathfrak{Re}\left\{\bar{\mathbf{w}}_{u}^{H}\mathbf{Y}_{u,q}\mathbf{Y}_{u,q}^{H}\mathbf{w}_{u}\right\}-\|\bar{\mathbf{w}}_{u}^{H}\mathbf{Y}_{u,q}\|^{2}$.
Since $\boldsymbol{\mu}_{u,k}\geq \mathbf{0}$ and $\xi_{u,j} \geq 0$, the tight upper bound of $\{\zeta_{u,u',k}^{\mathrm{C}}\}$ and $\{\zeta_{u,j}^{\mathrm{W}}\}$ can be obtained as
\begin{equation}\label{zeta}
\small{
\widehat{\zeta}_{u,q} =
\left\{
  \begin{array}{l}
    (a_{u,q}+\rho\lambda_{u,q})
    \bigg[2\mathfrak{Re}\left\{\bar{\mathbf{w}}_{u}^{H}\mathbf{Y}_{u,q}\mathbf{Y}_{u,q}^{H}\mathbf{w}_{u}\right\} \\ ~~-\|\bar{\mathbf{w}}_{u}^{H}\mathbf{Y}_{u,q}\|^{2}\bigg], ~~~~ \hbox{$\lambda_{u,q} > 0$,} \\
    a_{u,q}\bigg[2\mathfrak{Re}\left\{\bar{\mathbf{w}}_{u}^{H}\mathbf{Y}_{u,q}\mathbf{Y}_{u,q}^{H}\mathbf{w}_{u}\right\} -\|\bar{\mathbf{w}}_{u}^{H}\mathbf{Y}_{u,q}\|^{2}\bigg]\\
    ~~+\rho\lambda_{u,q}\|\mathbf{w}_{u}^{H}\mathbf{Y}_{u,q}\|^{2},
    ~~~~ \hbox{$\lambda_{u,q}\leq 0$. }
  \end{array}
\right.}
\end{equation}
After some rearrangement, we can obtain \eqref{ApproximatedZeta}.
This ends the proof.
\end{proof}

\subsection{Joint Problem Optimization}
After transformation, \eqref{ApproximatedAL} has been reduced to a multi-convex problem over $\boldsymbol{\omega}$ and $\boldsymbol{\widehat{\omega}}$, enabling the solution to be derived by decomposing \eqref{ApproximatedAL} into multiple sequential subproblems with BCU method.
Considering the block structure, primal and dual variables are partitioned into five blocks: $\{\mathbf{X}, \boldsymbol{\gamma}\}$, $\{\mathbf{p}\}$, $\{\mathbf{F}\}$, $\{\mathbf{W}, \mathbf{D}\}$, and $ \{\boldsymbol{\mu},\boldsymbol{\xi},\mathbf{\widetilde{X}}\}$,
each of which is alternatively solved in an individual subproblem by fixing the others.
Specifically, the first subproblem solves user grouping and SINR allocation, and the second one optimizes over power control.
Furthermore, the analog beamforming is designed to approach the performance of fully digital beamforming, followed by which digital beamforming matrixes are optimized.
In the last subproblem, the introduced auxiliary and dual variables are updated in parallel.

By fixing the remaining variables, the user grouping and SINR allocation subproblem optimizing over the variable block $\{ \mathbf{X}, \boldsymbol{\gamma}\}$ can be rearranged by
\begin{equation}\label{PU}
\small{
\begin{split}
\max_{\mathbf{X},\boldsymbol{\gamma}}~
&\frac{1}{2\rho} \sum_{k\in\mathcal{K}} \sum_{u\in\mathcal{U}} \left[|\mathbf{w}_{u}^{H}\widetilde{\mathbf{h}}_{u,k}|^{4}\left(p_{k}^{\mathrm{C}}\mathbf{c}_{u}^{T}\mathbf{x}_{k}\right)^{2}
- 2p_{k}^{\mathrm{C}}\widehat{\boldsymbol{\zeta}}_{u,k}^{\mathrm{C}}\mathbf{x}_{k}\right]
\\&+\frac{1}{2\rho}\sum_{k\in\mathcal{K}}\sum_{u\in\mathcal{U}}\left(\mathbf{c}_{u}^{T}\mathbf{x}_{k}-\widetilde{x}_{u,k}+\rho\lambda_{u,k}^{\mathrm{X}}\right)^{2}
\\& +\frac{1}{2\rho}\sum_{k\in\mathcal{K}}\sum_{u\in\mathcal{U}}\left[\mathbf{c}_{u}^{T}\mathbf{x}_{k}(1-\widetilde{x}_{u,k})+\rho\widetilde{\lambda}_{u,k}^{\mathrm{X}}\right]^{2}
\\&+ \frac{1}{2\rho}\sum_{k\in\mathcal{K}}\left(\sum_{u\in\mathcal{U}}\mathbf{c}_{u}^{T}\mathbf{x}_{k}-1+\rho\lambda_{k}^{\mathrm{U}}\right)^{2}
\\&   + \! \sum_{k\in\mathcal{K}} \sum_{u\in\mathcal{U}} \bigg[ \!
    \left(V_{0}(Q_{k}+Z_{k}^{\mathrm{C}})\tau+1\right) \log_2(1+\gamma_{u,k})
 \\&  - \! \frac{1}{2\rho} \Lambda_{u,k}^2 \gamma_{u,k}^2
     - \! \frac{1}{\rho} \Lambda_{u,k}(\rho\lambda_{u,k}^{\gamma} - \!\mu_{u,u,k}) \gamma_{u,k} \bigg],
\end{split}}
\end{equation}
subject to \eqref{MinRate_P3} and \eqref{Power_P3}.
In \eqref{PU}, $\Lambda_{u,k} = \sum_{(u',k')\in \mathcal{S}_{u,k}} \mu_{u,u',k} + \sum_{j \in \mathcal{J}_{\mathrm{T}}} \xi_{u,j} + \sigma^2$ denotes the aggregate interference and noise suffered by mmWave device $k$ served by RF chain $u$.
The optimal user grouping and SINR can be obtained with the following theorem.
\begin{theorem}\label{Theorem_U}
The optimal user grouping and SINR in \eqref{PU} can be achieved by
\begin{equation}\label{Sol_U}
\small{
\begin{split}
\mathbf{x}_{k}^{*}=& \bigg[\sum_{u\in\mathcal{U}}\left(|\mathbf{w}_{u}^{H}\widetilde{\mathbf{h}}_{u,k}|^{4}\left(p_{k}^{\mathrm{C}}\right)^{2}
+\widetilde{x}_{u,k}^{2}-2\widetilde{x}_{u,k}\right)\mathbf{c}_{u}\mathbf{c}_{u}^{T}
\\& +3\mathbf{I}_{U\times U}\bigg]^{-1} \times
\sum_{u\in\mathcal{U}}\left[p_{k}^{\mathrm{C}}\left(\widehat{\boldsymbol{\zeta}}_{u,k}^{\mathrm{C}}\right)^{T} \!+ A_{u,k}^{\mathrm{X}}\mathbf{c}_{u}\right],
\end{split}}
\end{equation}
\begin{equation}\label{Sol_SINR}\small{
\gamma_{u,k}^{*}\!=\!\frac{\sqrt{(E_{u,k}\!-\!\Lambda_{u,k}^{2})^{2}+4\widetilde{Q}_{u,k}\Lambda_{u,k}^{2}}-(E_{u,k}+\Lambda_{u,k}^{2})}{2\Lambda_{u,k}^{2}},}
\end{equation}
\noindent where $A_{u,k}^{\mathrm{X}} = \left(\widetilde{x}_{u,k}-\rho\lambda_{u,k}^{\mathrm{X}}\right) + \left(1-\rho\lambda_{k}^{\mathrm{U}}\right) + \left(\widetilde{x}_{u,k}-1\right) \rho\widetilde{\lambda}_{u,k}^{\mathrm{X}}-\rho\sum_{u}\kappa_{u,k}^{\mathrm{X(1)}}R_{\min}+\rho\kappa_{k}^{\mathrm{X}(2)}P^{\max}$,
and $\kappa_{u,k}^{\mathrm{X}(1)}$ and $\kappa_{k}^{\mathrm{X}(2)}$ are the Lagrangian multipliers for constraints \eqref{MinRate_P3} and \eqref{Power_P3}, respectively.
Moreover, $\widetilde{Q}_{u,k} = \frac{\rho}{\ln2}\big[V_0\big(Q_{k}+Z_{k}^{\mathrm{C}}\big)\tau + 1 +\kappa_{u,k}^{\gamma(1)}\big]$
and $E_{u,k} = \Lambda_{u,k}\big(\rho\lambda_{u,k}^{\gamma}-\mu_{u,u,k}\big) + \rho\kappa_{u,k}^{\gamma(2)}$.
Both $\kappa_{u,k}^{\gamma(1)}$ and $\kappa_{u,k}^{\gamma(2)}$ denote the Lagrangian multipliers for constraint \eqref{MinRate_P3}, which can be obtained based on the complementarity slackness condition of the associated constraints.
\end{theorem}

\begin{proof}
Since \eqref{PU} is a typical quadratic optimization problem with linear constraints with respect to $\mathbf{X}$, we can easily obtain \eqref{Sol_U} with the first-order optimality based on KKT conditions.
Moreover, we define the Lagrangian function respect to $\boldsymbol{\gamma}$ as
\begin{equation}
\small{
\begin{split}
\widehat{\mathcal{L}}_{\mathrm{\gamma}}=& \widehat{\mathcal{L}}(\boldsymbol{\gamma}) + \sum_{k\in\mathcal{K}}\left(\boldsymbol{\kappa}_k^{\gamma,(1)}\right)^{T}\left(\log_2(1+\boldsymbol{\gamma}_{k})-\mathbf{x}_{k}R^{\min}\right)
\\& - \left(\boldsymbol{\kappa}_k^{\gamma,(2)}\right)^{T}\left(\boldsymbol{\gamma_{k}}-\gamma_{k}^{\min}\right),
\end{split}}
\end{equation}
\noindent where $\boldsymbol{\kappa}_k^{\gamma,(1)}$ and $\boldsymbol{\kappa}_k^{\gamma,(2)}$ are the Lagrangian multiplier vectors.
By setting $\frac{\partial\widehat{\mathcal{L}}_{\mathrm{\gamma}}}{\partial\gamma_{u,k}}=0$, we have
\begin{multline}\small
\Lambda_{u,k}^{2}\gamma_{u,k}
+ \Lambda_{u,k} \left(\rho\lambda_{u,k}^{\gamma}-\mu_{u,u,k}\right)
+ \rho \kappa_{u,k}^{\gamma(2)}
\\- \frac{\rho}{\ln2}  \frac{V_0\left(Q_{k}+Z_{k}^{\mathrm{C}}\right)\tau + 1 +\kappa_{u,k}^{\gamma(1)}}{\gamma_{u,k}+1}
=0,
\end{multline}
\noindent which can be rewritten as
\begin{equation}
\small{
\Lambda_{u,k}^{2}\gamma_{u,k}^{2}+(E_{u,k}+\Lambda_{u,k}^{2})\gamma_{u,k}+E_{u,k}-\widetilde{Q}_{u,k} = 0.}
\end{equation}
Considering $\gamma_{u,k} \geq 0$, SINR can be obtained by \eqref{Sol_SINR}, which ends the proof.
\end{proof}


\begin{lemma}\label{SINRAlloc}
From \eqref{Sol_SINR}, we can easily derive that a NR-U user $k$, $\forall k \in \mathcal{K}$, is grouped into beam $u$ with $\gamma_{u,k} > 0$ when
\begin{equation}\small
V_0 \! \left(Q_{k}+Z_{k}^{\mathrm{C}}\right)\tau  +\kappa_{u,k}^{\boldsymbol{\gamma},(1)} + 1 \!>\!
\left[\left(\lambda_{u,k}^{\gamma}\!-\!\frac{\mu_{u,u,k}}{\rho}\right)\Lambda_{u,k} \!+\! \kappa_{u,k}^{\boldsymbol{\gamma},(2)}\right] \! \ln 2.
\end{equation}
\end{lemma}

With fixed $\mathbf{X}$ and $\mathbf{W}$, the power control problem can be rewritten as
\begin{equation}\label{PP}
\small{
\begin{split}
\min_{\mathbf{p}{^\mathrm{C}}}~ & \frac{1}{2\rho}\sum_{k\in\mathcal{K}}\sum_{u\in\mathcal{U}}\sum_{u'\in\mathcal{U}}
    \bigg[x_{u,k}^2|\mathbf{w}_{u'}^{H}\widetilde{\mathbf{h}}_{u',k}|^{4}(p_{k}^{\mathrm{C}})^{2}
    -\!2x_{u,k}\widehat{\zeta}_{u',u,k}^{\mathrm{C}}p_{k}^{\mathrm{C}}\bigg]
\\&  + Z^{\mathrm{W}} \sum_{k\in\mathcal{K}} \sum_{j\in\mathcal{J}_{\mathrm{R}}}
    \|\mathbf{V}_{j}^{H}\mathbf{g}_{j,k}^{\mathrm{W}}\|^{2} p_{k}^{\mathrm{C}},
\end{split}}
\end{equation}
subject to \eqref{WiGigInterferenceMitigation_P3} and \eqref{Power_P3}.

After analyzing, we have the closed-form power control expression as
\begin{equation}\label{Sol_p}
\small{
p_{k}^{\mathrm{C}} =
\left\{
  \begin{array}{l}
    \frac
    {\rho\left[\kappa_{k}^{\mathrm{p},(2)}-\kappa_{k}^{\mathrm{p},(3)}
    -\sum_{j}(Z^{\mathrm{W}}+\kappa_{j}^{\mathrm{p},(1)})\|\mathbf{V}_{j}^{H}\mathbf{g}_{j,k}^{\mathrm{W}}\|^{2}\right]}
    {\sum_{u\in\mathcal{U}}x_{u,k}^{2}\left(\sum_{u'\in\mathcal{U}}|\mathbf{w}_{u'}^{H}\widetilde{\mathbf{h}}_{u',k}|^{4}\right)}\\
    ~~+\frac
    {\sum_{u\in\mathcal{U}}\sum_{u'\in\mathcal{U}}x_{u,k}\widehat{\zeta}_{u',u,k}^{\mathrm{C}}}
    {\sum_{u\in\mathcal{U}}x_{u,k}^{2}\left(\sum_{u'\in\mathcal{U}}|\mathbf{w}_{u'}^{H}\widetilde{\mathbf{h}}_{u',k}|^{4}\right)},
     \hbox{if $\sum_{u=1}^{U}x_{u,k} > 0$,} \\
     0,
    ~~ \hbox{otherwise,}
  \end{array}
\right.}
\end{equation}
where $\kappa_{k}^{\mathrm{p},(1)}$, $\kappa_{k}^{\mathrm{p},(2)}$ and $\kappa_{k}^{\mathrm{p},(3)}$ are the nonnegative Lagrangian multipliers associated with constraints \eqref{WiGigInterferenceMitigation_P3} and \eqref{ContinuousVar_P3}.

\begin{algorithm}[h]
\caption{Analog Beamforming Optimization Algorithm for Problem \eqref{PF} Based on BCU}
\label{Alg_ABF}
\begin{algorithmic}[1]
\Require{
    $\mathbf{C}_{i}=\sum_{u\in\mathcal{N}_{i}^{\mathrm{C}}}\left(\left(\mathbf{w}_{u}+\rho\boldsymbol{\lambda}_{u}^{\mathrm{W}}\right)\mathbf{d}_{u}^{H}\right)$,
     $\mathbf{\widetilde{D}}_{i}=\sum_{u\in\mathcal{N}_{i}^{\mathrm{C}}}\mathbf{d}_{u}\mathbf{d}_{u}^{H}$.}
\State Initialize $l_{\mathrm{ABF}} = 0$, $\mathbf{F}_{i}$, and $\mathbf{U}_{i} = \mathbf{F}_{i}\sum_{u\in\mathcal{N}_{i}^{\mathrm{C}}}\mathbf{d}_{u}\mathbf{d}_{u}^{H}$.
\Repeat
    \For {$i = 1,\cdots,I$, $m = 1,\cdots,M_{0}$, $n = 1,\cdots,N_{0}$}
    \State Set $\small{b: = \mathbf{F}_{i}(m,n) \widetilde{\mathbf{D}}_{i}(n,n) - \mathbf{U}_{i}(m,n) + \mathbf{C}_{i}(m,n)}$.
    \State Update $\mathbf{U}_{i}: = \mathbf{U}_{i} + \left(\frac{b}{|b|}-\mathbf{F}_{i}(m,n)\right) \mathbf{I}_{M_0\times M_0}(:,m) \mathbf{\widetilde{D}}_i(n,:)$.
    \State Update $\mathbf{F}_{i}(m,n): = \frac{b}{|b|}$.
    \EndFor
    \State Update $l_{\mathrm{ABF}}: = l_{\mathrm{ABF}} + 1$.
\Until{$l_{\mathrm{ABF}} \geq N_{\mathrm{ABF}}^{\max}$ or the difference between successive values  satisfies the termination criterion.}
\Ensure{Analog beamforming matrix $\mathbf{F} = [\mathbf{F}_{1},\mathbf{F}_{2},...,\mathbf{F}_{I}]$.}
\end{algorithmic}
\end{algorithm}

Furthermore, the analog beamforming subproblem can be rewritten as
\begin{equation*}\tag{\ref{PF}}
\small{
\begin{split}
\min_{\mathbf{F}_i}~ 
&\mathrm{Tr}\bigg(\mathbf{F}_i^{H}\mathbf{F}_i\sum_{u\in\mathcal{N}_{i}^{\mathrm{C}}}\mathbf{d}_{u}\mathbf{d}_{u}^{H}\bigg)
\\&-2\mathfrak{Re}\bigg\{ \mathrm{Tr}\bigg[\mathbf{F}_i^{H}\sum_{u\in\mathcal{N}_{i}^{\mathrm{C}}}\bigg(\bigg(\mathbf{w}_{u}+\rho\boldsymbol{\lambda}_{u}^{\mathrm{W}}\bigg)\mathbf{d}_{u}^{H}\bigg)\bigg]\bigg\},
\end{split}}
\end{equation*}
subject to
\begin{subequations}\label{PF}
\begin{equation}\label{AnalogBF_PF}\small
|\mathbf{F}_i(m,n)=1|,~ \forall m \in \{1,2,...,M_{0}\}, n \in \{1,2,...,N_{0}\}.
\end{equation}
\end{subequations}
Since the objective function \eqref{PF} is quadratic, leveraging the separability of the unit-modulus constraint \eqref{AnalogBF_PF}, we can apply BCU algorithm to recursively update the elements of analog beamforming matrixes.
Specifically, during each step, only one element of $\mathbf{F}_{i}$ is optimized while fixing the others.
The derivation is similar to \cite{Shi2018Spectral}, and the analog beamforming optimization algorithm is presented in Algorithm \ref{Alg_ABF}.


By fixing the other variables, the digital beamforming subproblem can be recast into
\begin{equation}\label{PD}
\small{
\begin{split}
\min_{\mathbf{W},\mathbf{D}}~&
\frac{1}{2\rho} \sum_{u \in \mathcal{U}} \left\|\mathbf{w}_{u}-\mathbf{F}_{I(u)}\mathbf{d}_{u}+\rho\boldsymbol{\lambda}^{\mathrm{W}}\right\|^2
\\&+\frac{1}{2\rho} \sum_{k\in\mathcal{K}}  \sum_{u\in\mathcal{U}} \bigg[ \|\mathbf{x}_{k}\|^2 \left(p_{k}^{\mathrm{C}}\right)^{2} \left|\mathbf{w}_{u}^H \widetilde{\mathbf{h}}_{u,k}\right|^4
    - 2  \boldsymbol{\widehat{\zeta}}_{u,k}^{\mathrm{C}}  \mathbf{x}_{k} p_{k}^{\mathrm{C}}\bigg]
\\& + \frac{1}{2\rho} \sum_{u\in\mathcal{U}} \sum_{j \in \mathcal{J}_{\mathrm{T}}}
    \bigg[\|\mathbf{w}_{u}^H \mathbf{\widetilde{G}}_{u,j}^{\mathrm{C}}\|^4 \left(p_{j}^{\mathrm{W}}\right)^2
    - 2 p_{j}^{\mathrm{W}} \widehat{\zeta}_{u,j}^{\mathrm{W}}
    \bigg].
\end{split}}
\end{equation}
Moreover, $\boldsymbol{\mu}$ and $\boldsymbol{\xi}$ can be updated by
\begin{equation}\label{PSignal}
\small{
\begin{split}
\min_{\boldsymbol{\mu},\boldsymbol{\xi}}~&
\frac{1}{2\rho}\sum_{k\in\mathcal{K}} \sum_{u\in\mathcal{U}}
    \left[ \|\boldsymbol{\mu}_{u,k}\|^{2} + 2 \mathbf{S}_{u,k}^{\mathrm{C}} \boldsymbol{\mu}_{u,k}^T \right]
\\&+ \frac{1}{2\rho} \sum_{u\in\mathcal{U}} \sum_{j \in \mathcal{J}_{\mathrm{T}}}
   \left[ \xi_{u,j}^{2} + 2S_{u,j}^{\mathrm{W}}\xi_{u,j} \right]
\\&
+ \frac{1}{2\rho} \sum_{k\in\mathcal{K}} \sum_{u\in\mathcal{U}}
    \bigg[ \gamma_{u,k} \bigg( \sum_{(u',k') \in \mathcal{S}_{u,k}} \mu_{u,u',k'}
    + \sum_{j \in \mathcal{J}_{\mathrm{T}}} \xi_{u,j}  + {\sigma}^2 \bigg)
 \\&   - \mu_{u,u,k}
    + \rho\lambda_{u,k}^{\gamma} \bigg]^2,
\end{split}}
\end{equation}
where $\mathbf{S}_{u,k}^{\mathrm{C}} \!= \rho\boldsymbol{\lambda}_{u,k}^{\mu} \!+ \mathbf{x}_{k}^{T} p_{k}^{\mathrm{C}} \!\big[ |\bar{\mathbf{w}}_{u}^{H} \widetilde{\mathbf{h}}_{u,k}|^{2} \!- 2 \mathfrak{Re} \left\{\bar{\mathbf{w}}_{u}^{H} \widetilde{\mathbf{h}}_{u,k} \widetilde{\mathbf{h}}_{u,k}^{H} \mathbf{w}_{u}\right\}\!\big] $,
and $S_{u,j}^{\mathrm{W}} \!= \rho\lambda_{u,j}^{\xi} + p_{j}^{\mathrm{W}} \!\big[ \!\|\bar{\mathbf{w}}_{u}^{H}  \mathbf{\widetilde{G}}_{u,j}^{\mathrm{C}} \|^{2}   2\mathfrak{Re}\!\left\{\bar{\mathbf{w}}_{u}^{H} \mathbf{\widetilde{G}}_{u,j}^{\mathrm{C}} \big(\mathbf{\widetilde{G}}_{u,j}^{\mathrm{C}}\big)^{H} \mathbf{w}_{u}\right\} \big]$.
Both \eqref{PD} and \eqref{PSignal} are standard convex problems, which can be solved by conventional convex tools, e.g., the interior point method.
In addition, the dual variable $\mathbf{\widetilde{X}}$ is updated by solving the following quadratic programming under fixed $\mathbf{X}$ as
\begin{equation*}
\small{
\begin{split}
&\min_{\mathbf{\widetilde{X}}}
\frac{1}{2\rho}\sum_{u\in\mathcal{U}}\sum_{k\in\mathcal{K}}\left(1+x_{u,k}^{2}\right)\widetilde{x}_{u,k}^{2}
\\& -\sum_{u\in\mathcal{U}}\sum_{k\in\mathcal{K}}
\left[\frac{1}{\rho}x_{u,k}\left(x_{u,k}+1\right)+\lambda_{u,k}^{\mathrm{X}}+\widetilde{\lambda}_{u,k}^{\mathrm{X}}x_{u,k}\right]\widetilde{x}_{u,k},
\end{split}}\tag{\ref{PDualU}}
\end{equation*}
subject to
\begin{subequations}\label{PDualU}
\begin{equation}\label{EquvX2_PDualU}\small
\mathbf{0} \leq \mathbf{\widetilde{x}}_{k} \leq \mathbf{1}, ~\forall k \in \mathcal{K}.
\end{equation}
\end{subequations}

By taking the first-order optimality condition, we can obtain the unconstrained solution as
\begin{equation}\label{Sol_Upsilon}\small
\widetilde{x}_{u,k}^{\mathrm{o}}=
\frac{2\left[x_{u,k}\left(x_{u,k}+1\right)+\rho\left(\lambda_{u,k}^{\mathrm{X}}+\widetilde{\lambda}_{u,k}^{\mathrm{X}}x_{u,k}\right)\right]}
{1+x_{u,k}^{2}}.
\end{equation}

Considering \eqref{EquvX2_PDualU}, we have
\begin{align}\label{Sol_Upsilon2}
\small{
\widetilde{x}_{u,k}^{*}=
\left\lbrace
\begin{array}{ll}
1,& 1\leq\widetilde{x}_{u,k}^{\mathrm{o}},\\
\widetilde{x}_{u,k}^{\mathrm{o}},& 0 < \widetilde{x}_{u,k}^{\mathrm{o}}<1,\\
0,& \widetilde{x}_{u,k}^{\mathrm{o}}\leq 0.
\end{array}
\right.}
\end{align}

\begin{spacing}{1.1}
\begin{algorithm}[h]
\caption{Joint User Grouping, Beam Coordination and Power Control Strategy Based on \emph{PDD-CCCP}}
\label{Alg_PDD}
\begin{algorithmic}[1]
\State Define accuracy tolerance for BCU and PDD as $\epsilon_1$ and $\epsilon_2$, and define the maximum iteration number of BCU as $N_{\mathrm{BCU}}^{\max}$.
Initialize $\boldsymbol{\omega}^0,\boldsymbol{\widehat{\omega}}^0$ with a feasible point, the dual variables $\boldsymbol{\lambda}^{(0)}$, $\rho^{(0)}>0$, $0<\eta<1$, and $\delta^{(0)}>0$.
Set the iteration number of outer loop as $l=0$.
\Repeat
    \State Initialize the inner loop iteration number $n = 0$.
    \Repeat 
        \State Update user grouping and SINR allocation $\{\mathbf{X},\boldsymbol{\gamma}\}$ by solving \eqref{PU} with fixed $\{{\mathbf{D},\mathbf{p},\mathbf{\widetilde{x}},\boldsymbol{\mu},\boldsymbol{\xi}}\}$.
        \State Update power control $\mathbf{p}^{\mathrm{C}}$ by solving \eqref{PP} with fixed $\{\mathbf{X},\boldsymbol{\mu},\boldsymbol{\xi}\}$.
        \State Update analog beamforming $\mathbf{F}$ by solving \eqref{PF} with Algorithm \ref{Alg_ABF}.
        \State Update fully digital beamforming $\mathbf{W}$ and digital beamforming $\mathbf{D}$ by solving \eqref{PD}.
        \State Update $\boldsymbol{\mu}$, and $\boldsymbol{\xi}$ and $\mathbf{\widetilde{X}}$ by solving subproblem \eqref{PSignal} and \eqref{PDualU} in parallel.
        \State Update the inner loop iteration number: $n = n+1$.
    \Until{$n \geq N_{\mathrm{BCU}}^{\max}$ or  $|\hat{\mathcal{L}}^{(n+1)}(\boldsymbol{\omega},\boldsymbol{\widehat{\omega}})-\hat{\mathcal{L}}^{(n)}(\boldsymbol{\omega},\boldsymbol{\widehat{\omega}})| \leq \epsilon_1$.}
\State Calculate the constraint violation $h(\boldsymbol{\omega}^{(l)},\boldsymbol{\widehat{\omega}}^{(l)})$ according to \eqref{Violation}.
\If{$h(\boldsymbol{\omega}^{(l)},\boldsymbol{\widehat{\omega}}^{(l)}) \leq \delta^{(l)}$}
\State Update dual variables according to \eqref{DualUpdateX} - \eqref{DualUpdateSINR}. Set $\rho^{(l+1)} = \rho^{(l)}$.
\Else
\State Set $\boldsymbol{\lambda}^{(l+1)} = \boldsymbol{\lambda}^{(l)}$ and update penalty parameter $\rho^{(l+1)} = \eta\rho^{(l)}$.
\EndIf
\State Set $\delta^{(l+1)} = 0.9 h(\boldsymbol{\omega}^{(l)},\boldsymbol{\widehat{\omega}}^{(l)})$.
\State Update the outer loop iteration number: $l = l+1$.
\Until{$h(\boldsymbol{\omega}^{i},\boldsymbol{\widehat{\omega}}^{i})) < \epsilon_2$.}
\end{algorithmic}
\end{algorithm}
\end{spacing}

In the outer loop, the Lagrangian multipliers can be updated in the $l$-th iteration by
\begin{equation}\label{DualUpdateX}\small
\boldsymbol{\lambda}_{k}^{\mathrm{X},(l+1)}= \boldsymbol{\lambda}_{k}^{\mathrm{X},(l)} + \frac{1}{\rho} \left(\mathbf{x}_k- \mathbf{\widetilde{x}}_k\right),
\end{equation}
\begin{equation}\label{DualUpdateX2}\small
\widetilde{\lambda}_{u,k}^{\mathrm{X},(l+1)}= \widetilde{\lambda}_{u,k}^{\mathrm{X},(l+1)} + \frac{1}{\rho} x_{u,k}\left(1-\widetilde{x}_{u,k}\right),
\end{equation}
\begin{equation}\label{DualUpdateX3}\small
\widetilde{\lambda}_{k}^{\mathrm{U},(l+1)}= \widetilde{\lambda}_{k}^{\mathrm{U},(l+1)} + \frac{1}{\rho} \left(\sum_{u\in\mathcal{U}}\mathbf{c}_{u}^{T}\mathbf{x}_{k}-1\right),
\end{equation}
\begin{equation}\label{DualUpdateW}\small
\boldsymbol{\lambda}_{u}^{\mathrm{W},(l+1)}= \boldsymbol{\lambda}_{u}^{\mathrm{W},(l)} + \frac{1}{\rho} \left( \mathbf{w}_{u} - \mathbf{F}_{I(u)}\mathbf{d}_{u}\right),
\end{equation}
\begin{equation}\label{DualUpdateMu}\small
\boldsymbol{\lambda}_{u,k}^{\mu,(l+1)} = \boldsymbol{\lambda}_{u,k}^{\mu,(l)} + \frac{1}{\rho} \left(\boldsymbol{\mu}_{u,k} - \mathbf{x}_{k}^{T} |\mathbf{w}_{u}^{H} \widetilde{\mathbf{h}}_{u,k}|^2 p_{k}^{\mathrm{C}}\right),
\end{equation}
\begin{equation}\label{DualUpdateXi}\small
\lambda_{u,j}^{\xi,(l+1)} = \lambda_{u,j}^{\xi,(l)} + \frac{1}{\rho}\left(\xi_{u,j} - \|\mathbf{w}_{u}^{H} \widetilde{\mathbf{G}}_{u,j}^{\mathrm{C}}\|^2\right) ,
\end{equation}
\begin{equation}\label{DualUpdateSINR}\small
\lambda_{u,k}^{\gamma,(l+1)} = \lambda_{u,k}^{\gamma,(l)}+
\frac{1}{\rho}\bigg(\gamma_{u,k}\Lambda_{u,k} -\mu_{u,u,k}\bigg),
\end{equation}
\noindent where $\rho$ is the penalty parameter at current iteration, and $\lambda^{(l)}$ denotes the $l$-th iteration point of $\lambda$.

The detailed iterative algorithm is presented in Algorithm \ref{Alg_PDD}, where the constraint violation is defined as
\begin{equation}\label{Violation}
\small{
\begin{split}
 & h(\boldsymbol{\omega},\boldsymbol{\widehat{\omega}})
=  \max_{\forall u \in \mathcal{U}, \forall k \in \mathcal{K}}
\bigg\{
\left\|\mathbf{x}_{k} - \mathbf{\widetilde{x}}_{k}\right\|_{\infty},
|\mathbf{c}_{u}^{T}\mathbf{x}_{k}(1 - \widetilde{x}_{u,k})|,
 \\& \left|\mathbf{1}_{1\times U}\mathbf{x}_{k} \!- \! 1\right|,
\| \mathbf{w}_{u} \!-\! \mathbf{F}_{I(u)}\mathbf{d}_{u} \|,
\left\| \boldsymbol{\mu}_{u,k} \!-\! \mathbf{x}_{k}^{T} |\mathbf{w}_{u}^{H} \widetilde{\mathbf{h}}_{u,k}|^2 p_{k}^{\mathrm{C}} \right\|_{\infty},
 \\& \left| \xi_{u,j} - \|\mathbf{w}_{u}^{H} \widetilde{\mathbf{G}}_{u,j}^{\mathrm{C}}\|^2 \right|,
\left| \gamma_{u,k} \Lambda_{u,k} -\mu_{u,u,k} \right|
 \bigg\}.
\end{split}}
\end{equation}
Based on the convergence analysis of CCCP \cite{Lanckriet2009CCCP} and PDD \cite{Shi2017Penalty}, the proposed \emph{PDD-CCCP} algorithm can converge to a stationary solution of problem $\mathcal{P}3$.
Here, the computational complexity to solve $\mathbf{X}$ and $\boldsymbol{\gamma}$ is $O\left( K\left(U^3+UN_0+U^2\right) + U^2K^2\right)$.
Moreover, the computational for optimizing power control $\mathbf{p}^{\mathrm{C}}$ in (42) is $O\left(K(U^2+J_{\mathrm{R}})\right)$.
All the analog beamforming matrixes $\mathbf{F}_i, i \in \mathcal{I},$ are obtained with Algorithm 1, and the time complexity is $O\left(U^2M_0^2\right)$.
Additionally, the complexity for solving all the digital beamforming matrixes is $O\big(UN_{0}\big(M_{0}N_{0}+KU+\sum_{j\in\mathcal{J}_{\mathrm{T}}}M_{j}\big)\big)$, and the complexity for optimizing $\boldsymbol{\mu}$ and $\boldsymbol{\xi}$ is $O\left(U^2K(UK+M_0N_0+J_\mathrm{T})\right)$.
Based on the above analyses, the complexity of the developed \textit{PDD-CCCP} algorithm can be written as $O\big(U^{3}K^{2}+U^{2}M_{0}^{2}+U^{2}KM_{0}N_{0}+UM_{0}N_{0}^{2}+UN_{0}\sum_{j\in\mathcal{J}_{\mathrm{T}}}M_{j}\big)$.

\subsection{Performance Analyses}
The proposed algorithm based on Lyapunov optimization and \emph{PDD-CCCP} methods can achieve bounded NR-U traffic delay, limited long-term average WiGig interference, and asymptotically optimal NR-U spectral efficiency characterized by the following lemmas and theorems.
\begin{lemma}
If the virtual queue vectors $\mathbf{Z}(t)$ are upper bounded, the constraints \eqref{AvgWiGigInterference} and \eqref{Delay} can be satisfied.
\end{lemma}
\begin{proof}
Refer to \emph{Lemma 2} in \cite{Qiu2018Lyapunov}.
\end{proof}

\begin{theorem}\label{Theorem_QueueUpperBound}
The WiGig interference constraints \eqref{AvgWiGigInterference} and NR delay constraints \eqref{Delay} in $\mathcal{P}1$ can be guaranteed by Algorithm \ref{Alg_PDD}.
\end{theorem}
\begin{proof}
See Appendix \ref{Proof_QueueUpperBound}.
\end{proof}

\begin{theorem}\label{Theorem_UUpperBound}
The long-term average spectrum efficiency of mmWave NR-U network generated by optimizing $\mathcal{P}2$ is limited by a lower bound independent of the operation SP, expressed by
\begin{equation}\label{UUpperBound}\small
R^{\Pi}=\lim_{T\to+\infty}\frac{1}{T}\sum_{t=0}^{T-1}\mathbb{E}\left[R(t)\right]\geq R^{\mathrm{opt}}-\left(B_{1}+V_{0}\sum_{k\in\mathcal{K}}(Q_{k}^{\max})^2\right),
\end{equation}
\noindent where $Q_{k}^{\max}$ is the upper bound of $Q_{k}(t)$, and $R^{\mathrm{opt}}$ is the optimal long-term average mmWave spectrum efficiency under $\mathcal{P}1$.
\end{theorem}

\begin{proof}
See Appendix \ref{Proof_UUperBound}.
\end{proof}

From Theorem \ref{Theorem_UUpperBound}, it is demonstrated that the NR-U spectral efficiency $R^{\Pi}$  can be arbitrarily close to the optimal $R^{\mathrm{opt}}$ by selecting asymptotically small $V_0$ and $V_1$.


\section{Simulation}
In this section, we present the numerical results of the designed algorithm, and compare it with existing schemes from other literatures.
We evaluate the performance of the proposed coexistence strategy with $I = 2$ NR APs and $J_{\mathrm{A}} = 3$ WiGig APs together deployed in a $10 \times 10 ~ m^2$ coverage, operating at the same $60$ GHz band. The bandwidth is set as $20$ MHz.
We assume $20$ SPs in each DTI, and each SP has an equal transmission duration $\tau = 25.6$ ms \cite{SP_setting}.
Each NR AP equips $M_{0} = 128$ ULA antennas and $N_{0} = 8$ RF chains.
Moreover, each WiGig AP $j$ has $M_{j}^{\mathrm{A}} = 64$ and $N_{j}^{\mathrm{A}} = 4$ RF chains, while each WiGig user $j$ equips $M_{j}^{\mathrm{U}}= 4$ antennas and a single RF chain.
The transmit power of all the WiGig APs is defined as $p^{\mathrm{W}} = 10$ dBm, and the maximum transmit power of NR-U users is $P^{\max} = 10$ dBm.
Each WiGig AP serves $32$ WiGig users, which are divided into different WiGig MIMO user groups that are scheduled in different SPs.
With the trained MIMO user groups and hybrid beamforming configuration, intra-RAT interference in WiGig network has been suppressed.
Moreover, the maximal inter-RAT interference of each WiGig receiver is set as $I_{\max}^{\mathrm{W}} = -60$ dBm.
On the other hand, the traffic arrival of each NR user follows Poisson distribution with density $5$ bit/s/Hz.
The minimal rate requirements of NR users are set as $R_{\min} = 0.1$ bit/s/Hz,
and the delay requirements uniformly span from $1$ ms to $10$ ms.
The channel model parameters of user $k$ are set according to \cite{Dai2019HybridPrecodingMIMONOMA} with one LoS link and $L=3$ NLoS links, and the AoAs/AoDs uniformly distributed within $[0, 2\pi]$.
Moreover, the path loss model is set according to \cite{Lagen2018LBR}, and the noise power spectral density is $-134$ dBm/Hz.
For the \textit{PDD-CCCP} algorithm, we set the initial penalty factor $\rho^{(0)} = 0.5$, control parameter $\delta^{(0)} = 1$, the maximum inner loop number $N_{\mathrm{BCU}}^{\max} = 10$, and the threshold $\epsilon_{1} = 10^{-3}$, $\epsilon_{2} = 10^{-3}$.
The proposed algorithms are evaluated through $10000$ Monte-Carlo realizations.

In this simulation, we consider four baseline schemes, namely, \textit{CHS-HBF-FP},  \textit{CHS-HBF-EP}, \textit{CHS-HBF-AO}, and  \textit{CHS-BCU-SCA}.
In the above baseline schemes, we decide the user grouping and analog beamforming strategies based on channel correlation \cite{Dai2019HybridPrecodingMIMONOMA}.
Specifically, user grouping is first performed by extending the cluster head selection (CHS) algorithm in \cite{Dai2019HybridPrecodingMIMONOMA} to multi-cell scenario.
In addition, we adopt the hybrid beamforming (HBF) scheme, where each NR-U AP designs analog beamforming to maximize the effective gain of cluster head user, and mitigate inter-beam interference through zero-forcing digital beamforming.
In \textit{CHS-HBF-FP} and \textit{CHS-HBF-EP} schemes, the power coefficients are respectively determined based on fixed power (FP) control and equal power (EP) control, where $p_{k}^{\mathrm{C,FP}} = \frac{nP^{\max}}{\sum_{k\in\mathcal{K}}x_{u,k}}$ and $p_{k}^{\mathrm{C,EP}} = \frac{P^{\max}}{\sum_{k\in\mathcal{K}}}$.
In \textit{CHS-HBF-AO} scheme, a dynamic power control algorithm based on \cite{MIMONOMA_AO} is considered. Under \textit{CHS-HBF-AO}, the power control algorithm is solved with the auxiliary variables based on alternating optimization (AO).
Furthermore, we compare the proposed algorithm with \textit{CHS-BCU-SCA}, which is a joint digital beamforming and power control scheme. The \textit{CHS-BCU-SCA} tackles the non-convex minimal SINR constraint based on successive convex approximation (SCA), and then jointly solves the digital beamforming and power allocation problem based on BCU \cite{MIMO_BCU_SCA}.

\begin{figure}[htbp]
\centering
    \subfloat[Spectral efficiency of NR-U networks.]{\label{fig_convergenceR}
        \includegraphics[width=0.45\textwidth]{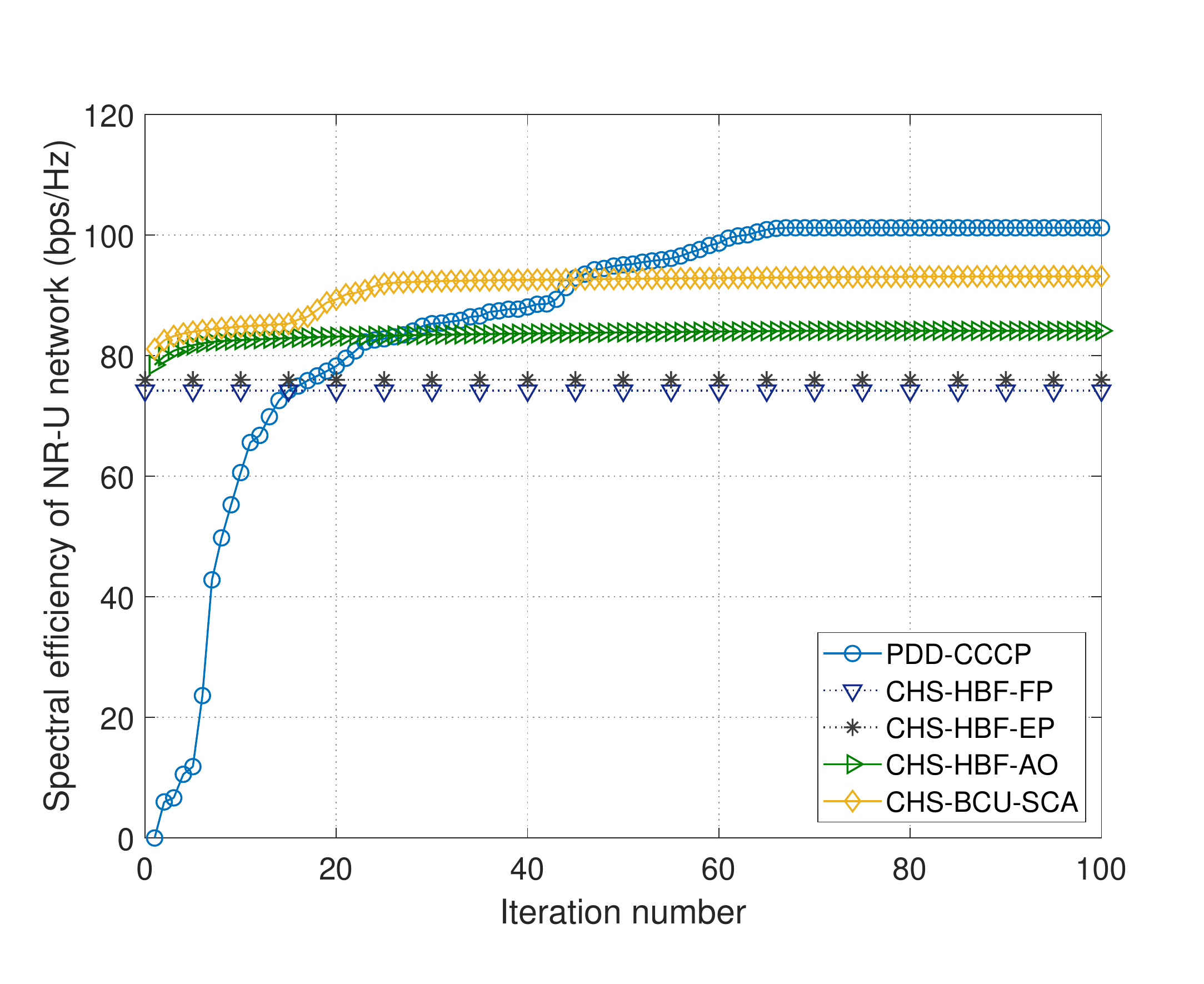}
    }
    \\
	\subfloat[Inter-RAT interference to WiGig.]{\label{fig_convergenceIW}
        \includegraphics[width=0.45\textwidth]{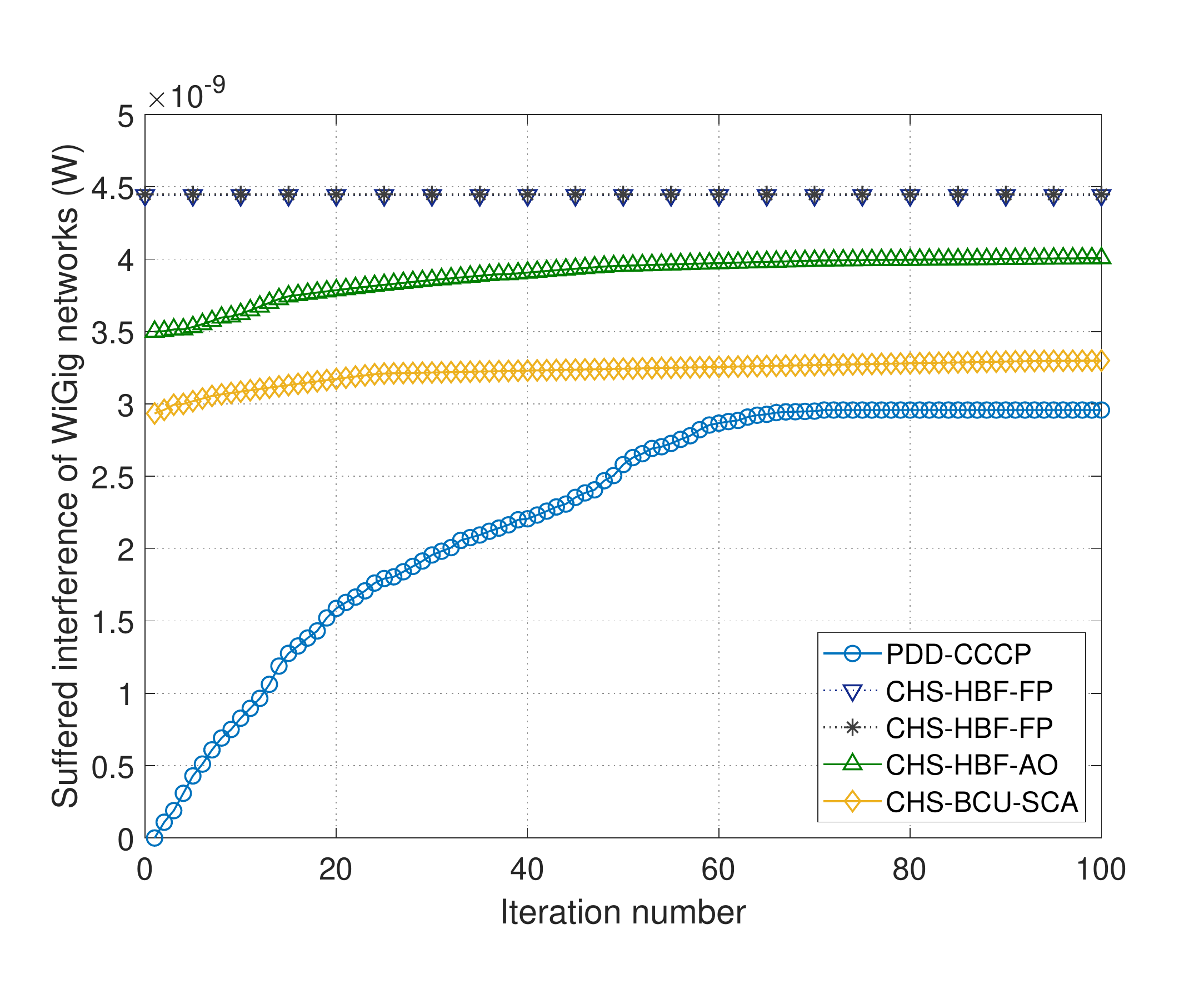}
    }
    \caption{Convergence process of the proposed \emph{PDD-CCCP} algorithm.}
    \label{fig_convergence}
\end{figure}

\begin{figure}
  \centering
  \includegraphics[width=0.45\textwidth]{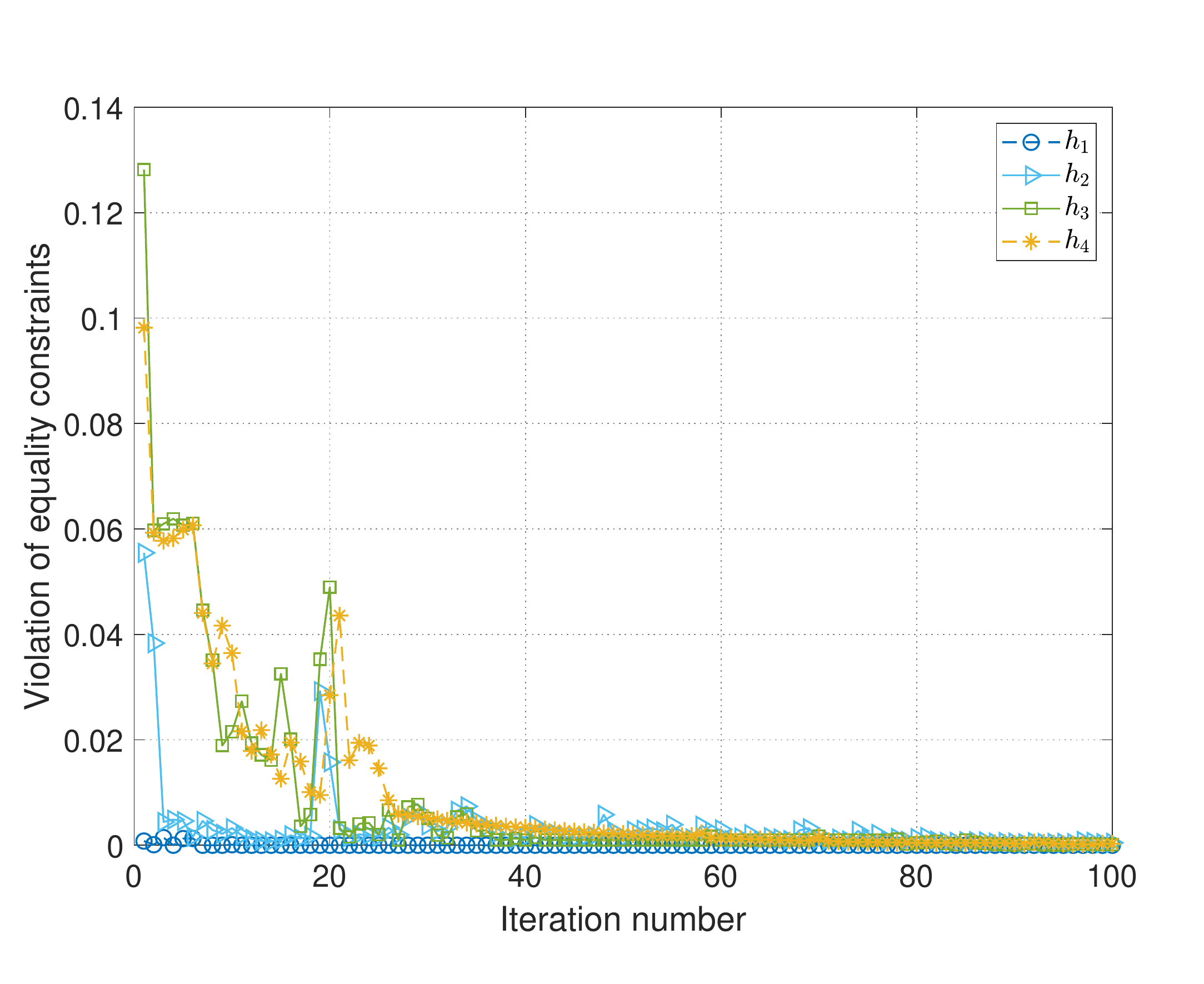} \\
  \caption{Validation of equivalent constraint transformations.}\label{fig_convergenceVio}
\end{figure}

The convergence performance of the proposed \emph{PDD-CCCP} algorithm is presented in Fig. \ref{fig_convergence}.
In Fig. \ref{fig_convergence}\subref{fig_convergenceR}, the achievable NR-U spectral efficiency converges within less than $100$ inner iterations.
In Fig. \ref{fig_convergence}\subref{fig_convergenceIW}, the aggregate inter-RAT interference suffered by WiGig users firstly increases with the increment of NR-U data rate, and then stabilizes at a certain level.
The proposed joint user grouping, hybrid beam coordination and power control algorithm outperforms both \textit{CHS-BCU-AO} and \textit{CHS-BCU-SCA} in terms of spectral efficiency and interference mitigation with a relatively slower convergence speed.

\begin{figure}[htbp]
\centering
\subfloat[NR-U spectral efficiency.]{\label{fig_R}
    \includegraphics[width=0.45\textwidth]{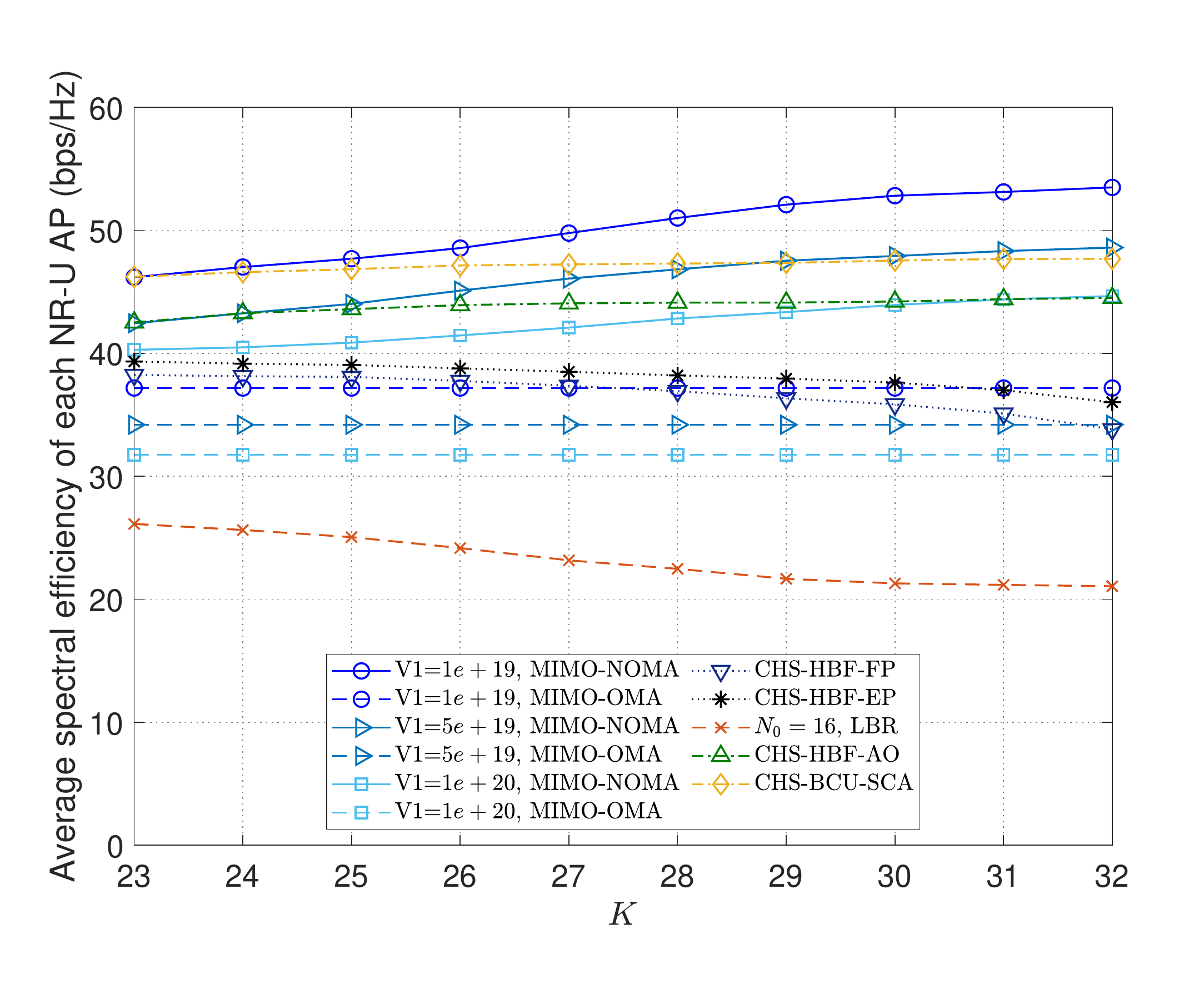}
}
\\
\subfloat[NR-U traffic delay.]{\label{fig_delay_K}
    \includegraphics[width=0.45\textwidth]{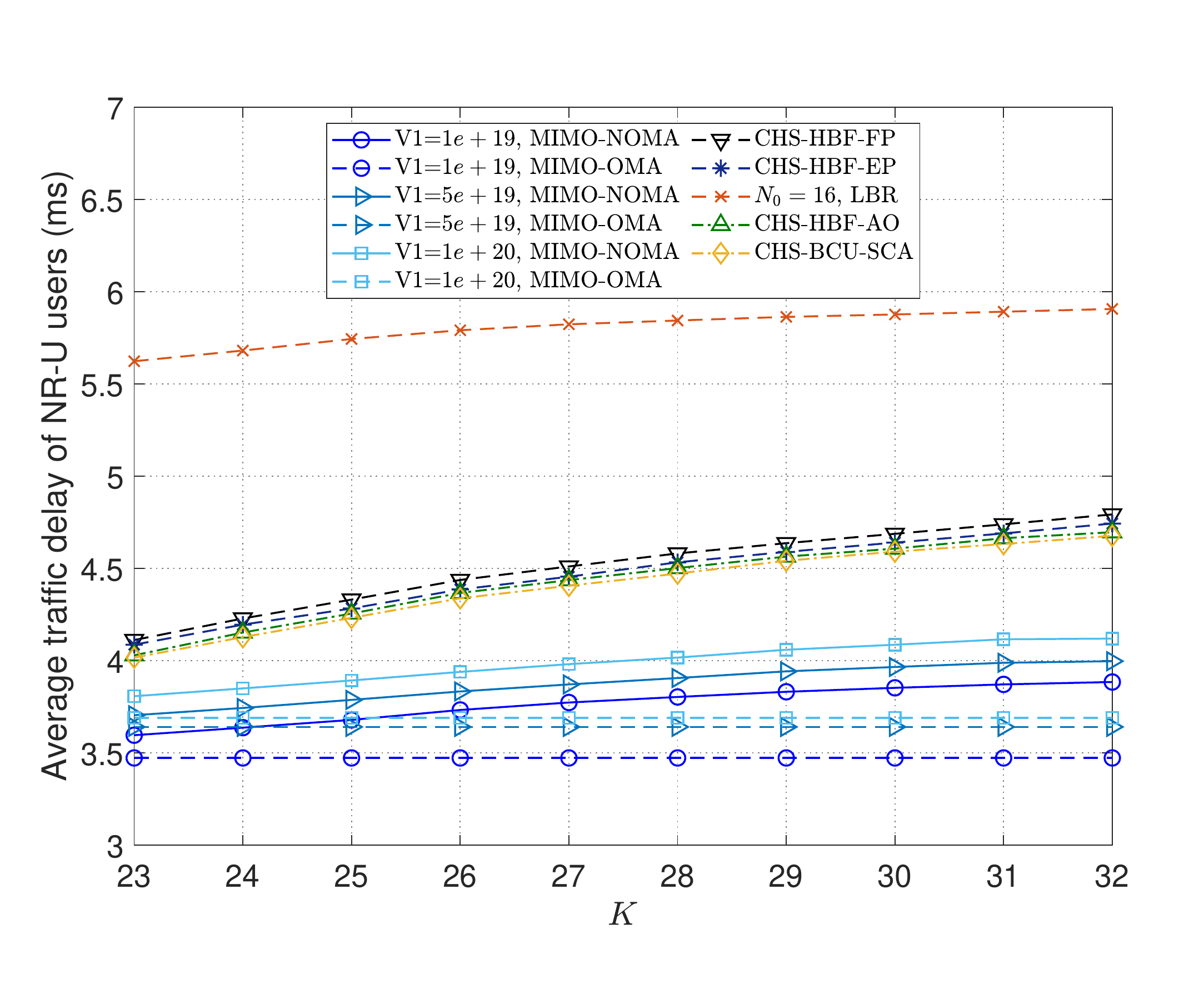}
}
\caption{NR-U performance comparison under various NR-U user number in different algorithms.}
\label{fig_NR-U}
\end{figure}

Fig. \ref{fig_convergenceVio} shows the validation of equivalent constraint transformations from $\mathcal{P}2$ to $\mathcal{P}3$.
Here, we denote constraint violations $h_{i}$, $\forall i \in \{1,2,3,4\}$ respectively represent the difference values of penalty terms from \eqref{AL1} to \eqref{AL4}.
From this figure, $h_{i}, \forall i \in \{1,2,3,4\}$ eventually descends to a low value approximating $0$, and thus the equality constraints in $\mathcal{P}3$ can be satisfied, namely $\mathcal{P}2$ and $\mathcal{P}3$ are equivalent.

\begin{figure}
  \centering
  \includegraphics[width=0.45\textwidth]{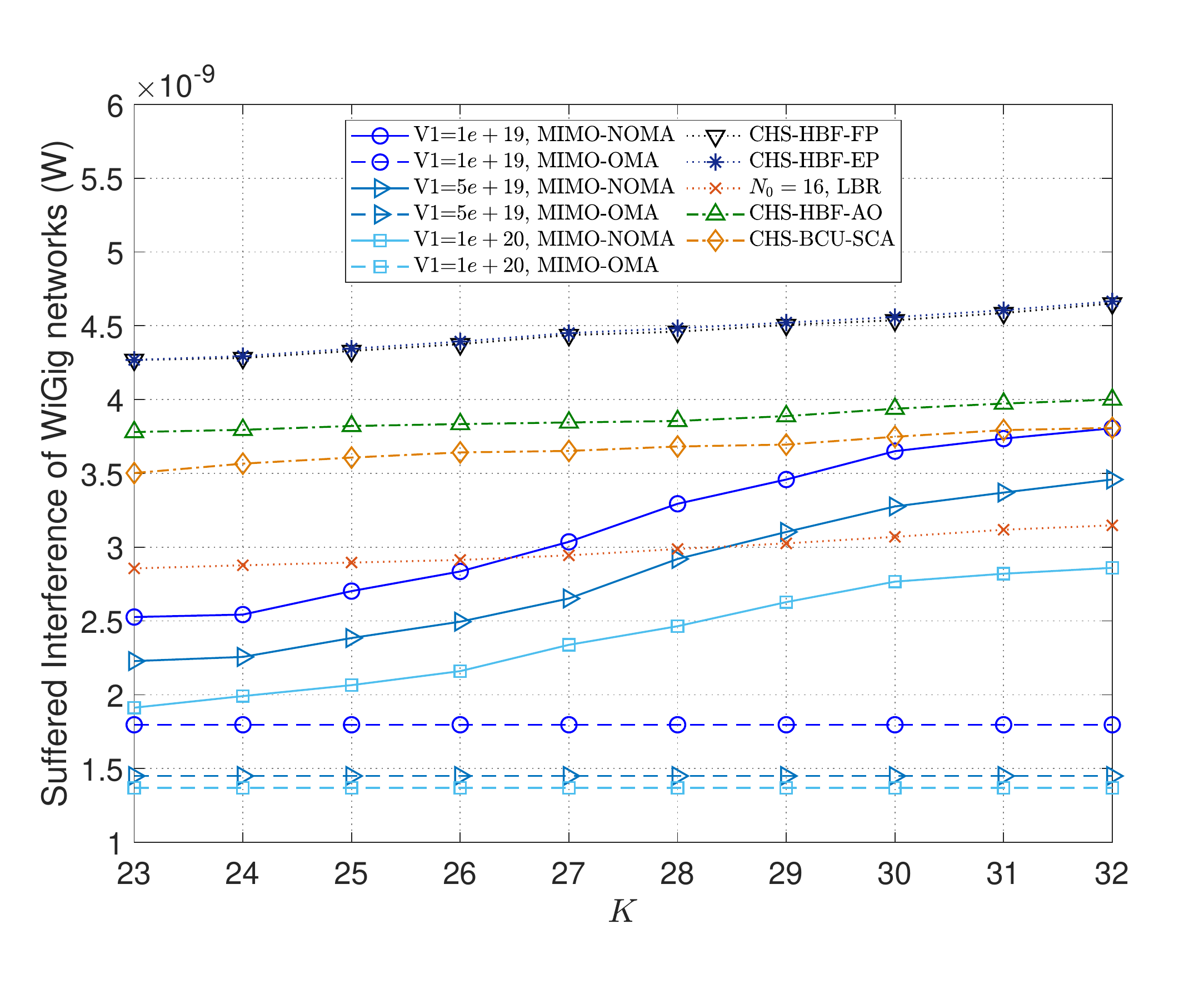} \\
  \caption{Inter-RAT interference to WiGig under various NR-U user number in different algorithms.}\label{fig_IW_K}
\end{figure}

\begin{figure}[htbp]
\centering
\subfloat[NR-U traffic delay.]{\label{fig_delay_V1}
    \includegraphics[width=0.45\textwidth]{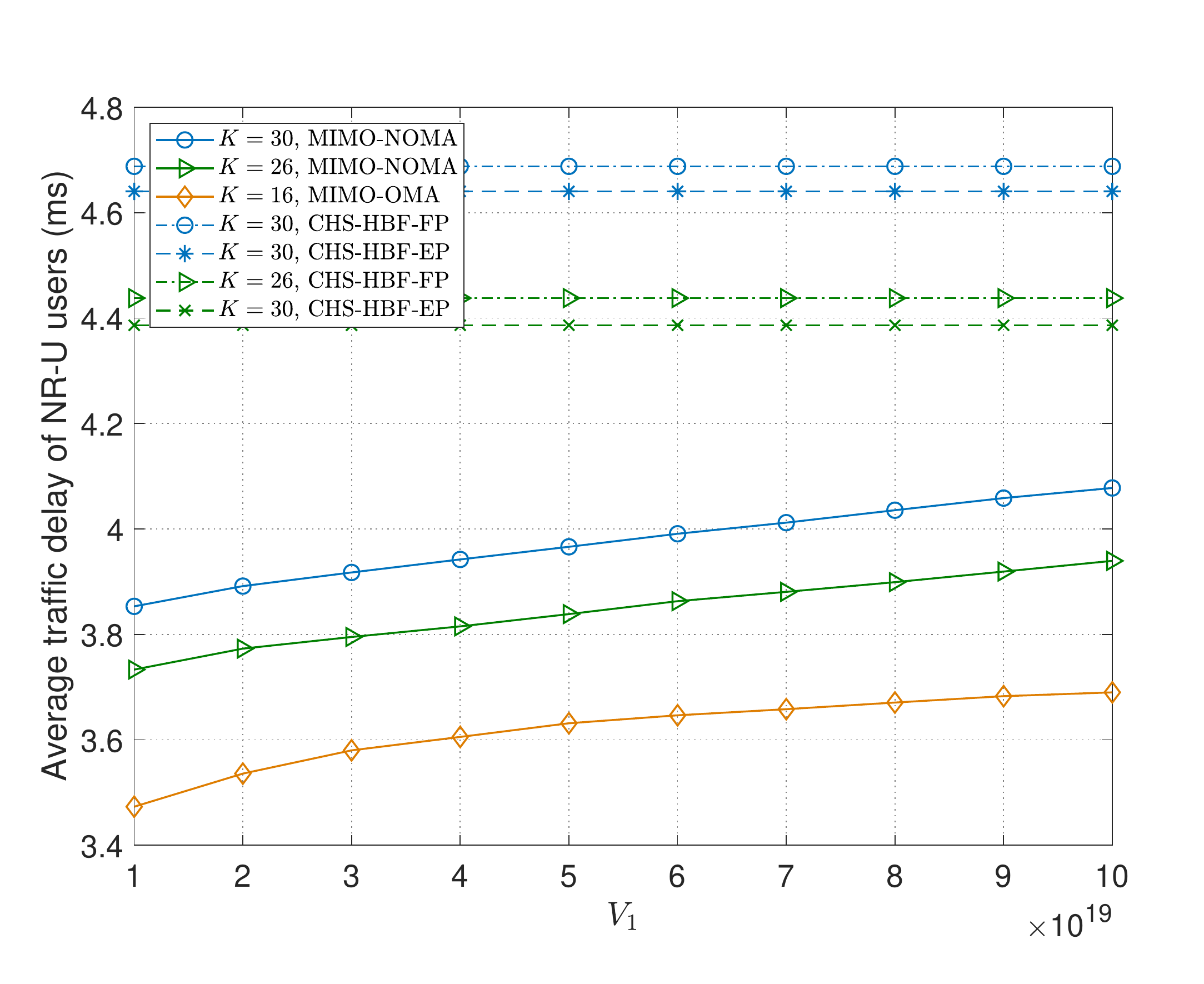}
}
\\
\subfloat[Inter-RAT interference to WiGig.]{\label{fig_delay_V1}
    \includegraphics[width=0.45\textwidth]{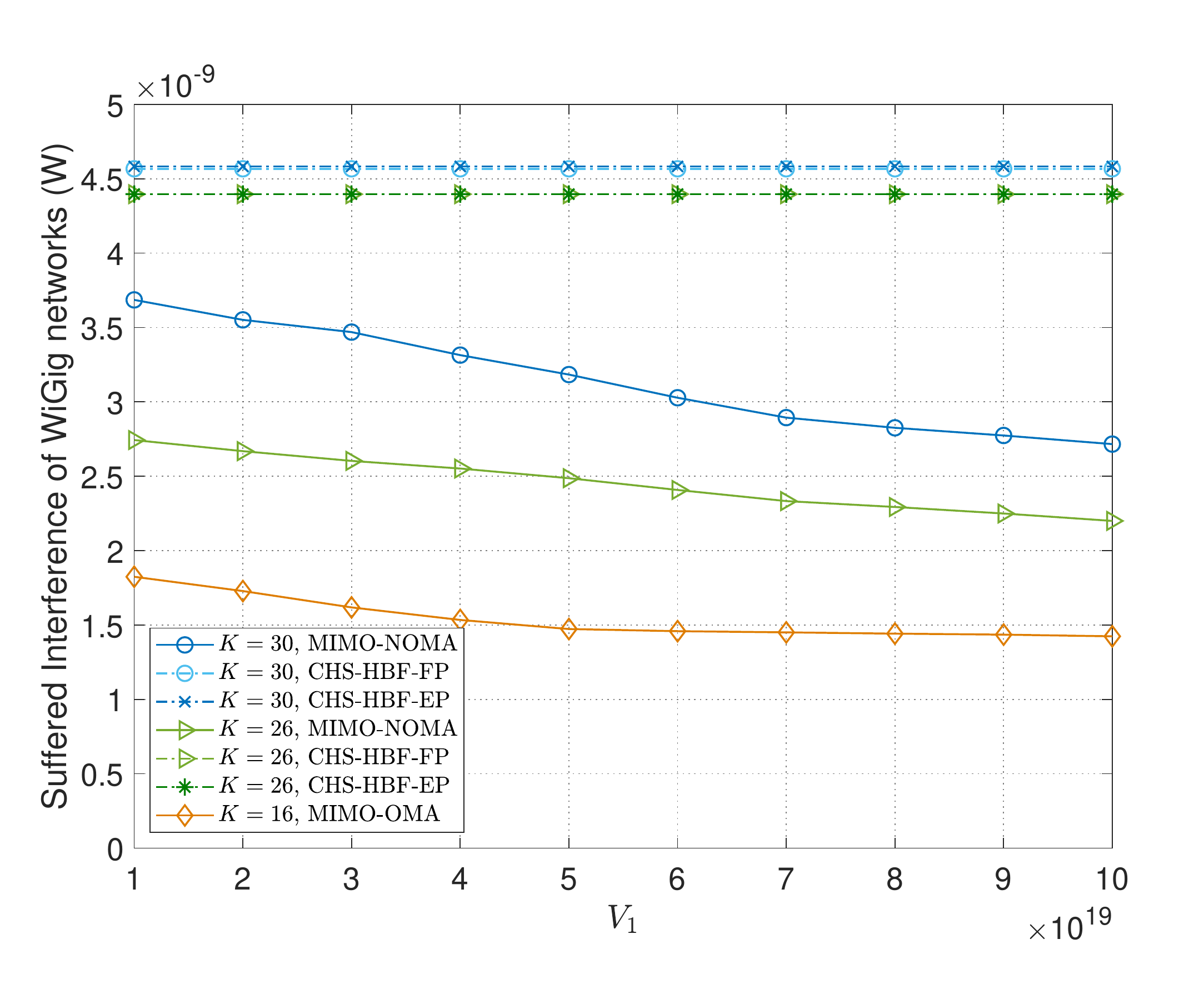}
}
\caption{Influence of Lyapunov control parameter $V_{1}$ under the different algorithms.}
\label{fig_V}
\end{figure}

In Fig. \ref{fig_NR-U}, we present the NR-U performance comparison under different algorithms. Here, performance of both the proposed schemes, the baseline schemes, and the exiting beam management LBR scheme are compared.
For the LBR scheme, we assume each NR-U AP equips $N_{0} = 16$ RF chains to support densely connected users, and the energy detection (ED) threshold is set as $-62$ dBm.
Moreover, we set the Lyapunov control parameter $V_{0} = 3\times10^{5}$.
From this figure, we can conclude that the proposed algorithm leads to the highest spectral efficiency and the lowest NR-U traffic delay under relatively small $V_1$ due to the joint user grouping, coordinated hybrid beamforming and power control design.
Meanwhile, the performance gaps increase with the number of NR-U users $K$.
Moreover, the LBR mechanism has the highest traffic delay and the lowest spectral efficiency since it prevents some of the users to access the wireless channels and it operates in an uncoordinated manner.
Furthermore, the spectral efficiency of NR-U under the \textit{CHS-HBF-FP} and \textit{CHS-HBF-EP} schemes in Fig. \ref{fig_R} decreases with the increment of NR-U users due to the enhanced inter-user interference.
Accordingly, Fig. \ref{fig_IW_K} shows the inter-RAT interference to WiGig under various number of NR-U users, which also shows the performance advantages of the proposed scheme.

In Fig. \ref{fig_V}, we illustrate the influence of Lyapunov control parameters $V_{1}$ under the proposed algorithms and parts of the baseline algorithms.
Here, we set the long-term average interference threshold $\bar{I}_{\max}^{\mathrm{W}} = -54$ dBm, and $V_{1}$ increases from $10^{19}$ to $10^{20}$.
From this figure, the interference suffered from WiGig users decreases with $V_{1}$ at the cost of higher NR-U traffic delay. This is because the average interference asymptotically reduces to the predefined threshold.
In addition, it also shows that the tradeoff between NR-U traffic delay and WiGig interference can be dynamically and adaptively controlled.

\section{Conclusion}
In this work, we introduce a novel uplink mmWave NR-U and WiGig coexistence framework, where multiple NR APs leverage MIMO-NOMA to coordinate with incumbent WiGig transmissions.
The investigated framework is capable to support intensively connected users, mitigating both intra-RAT and inter-RAT interference to enhance spectral utilization and achieve harmonious coexistence.
Through a joint user grouping, beam coordination and power control scheme, the spectral efficiency of the connected NR-U users are maximized while ensuring stringent NR-U delay requirement and the performance of WiGig users. The objective function is formulated as a Lyapunov optimization based MINLP problem with nonconvexity, discontinuity, and unit-modulus constraints. To deal with the proposed NP-hard problem, we propose a dual-loop \emph{PDD-CCCP} algorithm.  It first transfers the objective function into an AL problem, and then approximately decompose the AL problem into multiple sequential convex subproblems through CCCP and BCU methods.
Numerical results are presented to verify the effectiveness of the proposed algorithms.

\appendices
\section{Proof of the Lemma \ref{Lemma_DriftPenalty}}
Considering the definitions of $Q_{k}(t+1)$ and $Z_{k}^{\mathrm{C}}(t+1)$ and the fact $([x]^+)^2 \leq x^2$, we have
\begin{equation}\label{SquareZBound}\small{
\begin{split}
& \left[Z_{k}^{\mathrm{C}}(t+1)\right]^2
\overset{(a)}{\leq}  (Z_{k}^{\mathrm{C}}(t))^2 + (\bar{Q}_{k}) ^2  + (A_{k}^{\max})^2 + 2 Q_{k}^2(t)
\\& + 2 Z_{k}^{\mathrm{C}}(t)[Q_{k}(t) - \bar{Q}_{k} ]
- 2 [Q_{k}(t) + Z_{k}^{\mathrm{C}}(t)]R_{k}(t)\tau(t),
\end{split}}
\end{equation}
where $(a)$ is due to $ Q_{k}(t) - R_{k}(t) \tau(t) \geq 0$.

Substituting \eqref{SquareZBound} into definition of Lyapunov drift in \eqref{LyapunovDrift}, $\Delta Z_{k}^{\mathrm{C}}(t)$ is bounded by
\begin{equation}\label{DriftBound1}\small{
\begin{split}
& \Delta Z_{k}^{\mathrm{C}}(t)  = \frac{1}{2}[(Z_{k}^{\mathrm{C}})^2(t+1) - (Z_{k}^{\mathrm{C}}(t))^2] \\
& \leq \frac{1}{2}[(\bar{Q}_{k})^{2} \!+\! (\bar{A}_{k}^{\max})^{2}]+Q_{k}^{2}(t) - [Z_{k}^{\mathrm{C}}(t)+Q_{k}(t)]R_{k}(t)\tau(t) \\
& + Z_{k}^{\mathrm{C}}(t)[Q_{k}(t) - \bar{Q}_{k} ], \\
\end{split}}
\end{equation}
where the inequality can be obtained by combining constraint \eqref{MinRate}.

Similarly, we can obtain the upper bound of $\Delta Z^{\mathrm{W}}(t)$ as
\begin{equation}\label{DriftBound2}\small{
\begin{split}
  \Delta Z^{\mathrm{W}}(t) & \leq \frac{(\bar{I}^{\mathrm{W}}) ^2  + (I^{\mathrm{W}}(t))^2} {2} + I^{\mathrm{W}}(t) [  Z^{\mathrm{W}}(t) - \bar{I}^{\mathrm{W}}]  \\
    & \leq \frac{(\bar{I}^{\mathrm{W}}) ^2  + (|\mathcal{J}_{\mathrm{R}}| \gamma_{\mathrm{W}})^2} {2} + I^{\mathrm{W}}(t) [ Z^{\mathrm{W}}(t) - \bar{I}^{\mathrm{W}}] .
\end{split}}
\end{equation}

\section{Proof of Theorem \ref{Theorem_QueueUpperBound}}\label{Proof_QueueUpperBound}
Because $Z_{k}^{\mathrm{C}}(t)$ is updated according to data rate, we consider the solution of subproblem \eqref{PU} with respect to SINR $\boldsymbol{\gamma}$ to demonstrate $Z_{k}^{\mathrm{C}}(t)$ generated by algorithm \ref{Alg_PDD} is upper bounded.
Consider the Lagrangian function $\mathcal{\widehat{L}}_{\gamma}(\boldsymbol{\gamma})$ corresponding to \eqref{PU} with respect to $\boldsymbol{\gamma}$.
Since $\mathcal{\widehat{L}}_{\gamma}(\boldsymbol{\gamma})$ is concave over $\boldsymbol{\gamma}$, the first-order derivative should decrease with $\boldsymbol{\gamma}$.
Define $f^1(\boldsymbol{\gamma}) = \log_2(1+\boldsymbol{\gamma})-\mathcal{L}_4(\boldsymbol{\gamma}) + \sum_{k\in\mathcal{K}}\left(\boldsymbol{\kappa}_k^{\gamma,(1)}\right)^{T}\left(\log_2(1+\boldsymbol{\gamma}_{k})-\mathbf{x}_{k}R^{\min}\right) - \left(\boldsymbol{\kappa}_k^{\gamma,(2)}\right)^{T}\left(\boldsymbol{\gamma_{k}}-\gamma_{k}^{\min}\right)$,
and $f^{\mathrm{r}}(\gamma_{u,k}) = V_{0}\log_2(1+\gamma_{u,k})\tau$.
Denote the optimal SINR obtained by \eqref{PU} as $\boldsymbol{\gamma}^{*}$.
From constraint \eqref{MinRate_P3}, the optimal solution of SINR should satisfy $\boldsymbol{\gamma} \geq 0$.
Therefore, when $\frac{\partial \mathcal{\widehat{L}}_{\gamma}(\boldsymbol{\gamma})}{\partial \boldsymbol{\gamma}}|_{\boldsymbol{\gamma}=0} \geq 0$,
we have $\frac{\partial \mathcal{\widehat{L}}_{\gamma}(\boldsymbol{\gamma})}{\partial \boldsymbol{\gamma}}|_{\boldsymbol{\gamma}=\boldsymbol{\gamma}^*} \leq \frac{\partial \mathcal{\widehat{L}}_{\gamma}(\boldsymbol{\gamma})}{\partial \boldsymbol{\gamma}}|_{\boldsymbol{\gamma}=0}$, which implies
\begin{equation}\small{
\begin{split}
& \frac{\partial f^{1}(\gamma_{u,k})}{\partial \gamma_{u,k}} \bigg| _{\gamma_{u,k}=0}
+ (Q_{k}+Z_{k}^{\mathrm{C}}) \sum_{u \in \mathcal{U}} \frac{\partial f^{\mathrm{r}}(\gamma_{u,k})}{\partial \gamma_{u,k}} \bigg| _{\gamma_{u,k}=0}
\geq \\&
\frac{\partial f^{1}(\gamma_{u,k})}{\partial \gamma_{u,k}} \bigg| _{\gamma_{u,k} = \gamma_{u,k}^*}
+ (Q_{k}+Z_{k}^{\mathrm{C}}) \sum_{u \in \mathcal{U}} \frac{\partial f^{\mathrm{r}}(\gamma_{u,k})}{\partial \gamma_{u,k}} \bigg| _{\gamma_{u,k} = \gamma_{u,k}^*}\!,
 \forall u, k.
\end{split}}
\end{equation}

From the above inequality, the upper bound of $Z_{k}^{\mathrm{C}}$ can be expressed by
\begin{equation}\label{Z1Bound}\small
Z_{k}^{\mathrm{C}} (t) \leq \Omega_{k} - Q_{k} (t),
\end{equation}
where $\Omega_{k} = \frac
{\frac{\partial f^{1}(\gamma_{u,k})}{\partial \gamma_{u,k}} \bigg| _{\gamma_{u,k} = \gamma_{u,k}^*}
- \frac{\partial f^{1}(\gamma_{u,k})}{\partial \gamma_{u,k}} \bigg| _{\gamma_{u,k} = 0}}
{\sum_{n = 1}^{N_{0}} \bigg(
    \frac{\partial f_{u,k}^{\mathrm{r}}(\gamma_{u,k})}{\partial \gamma_{u,k}} \bigg| _{\gamma_{u,k} = \gamma_{u,k}^*} - \frac{\partial f_{u,k}^{\mathrm{r}}(\gamma_{u,k})}{\partial \gamma_{u,k}} \bigg| _{\gamma_{u,k} = 0} \bigg)}$.
Otherwise, when $\frac{\partial \mathcal{\widehat{L}}_{\gamma}(\boldsymbol{\gamma})}{\partial \boldsymbol{\gamma}}|_{\boldsymbol{\gamma}=0} \leq 0$, we can obtain $\boldsymbol{\gamma}^* = 0$.

By adding $[Q_{k}(t+1)-\bar{Q}_{k}]^+$ on both sides of \eqref{Z1Bound}, we have $Z_{k}^{\mathrm{C}} (t+1) \leq Z_{k}^{\mathrm{C}} (t) + [Q_{k}(t+1)-\bar{Q}_{k}]^+ \leq \Omega_{k} - Q_{k}(t) + [Q_{k}(t+1)-\bar{Q}_{k} (t)]^+$, $\forall k$.
From the above analyses, $Z_{k}^{\mathrm{C}}$ is upper bounded.
Similarly, considering the power control subproblem in \eqref{PP}, we can prove $Z^{\mathrm{W}}(t)$ is upper bounded.
Hence, both WiGig interference constraints and NR delay constraints can be guaranteed.

\section{Proof of Theorem \ref{Theorem_UUpperBound}}\label{Proof_UUperBound}
By substituting $B_2 = V_{0}\sum_{k\in\mathcal{K}} [ Q_{k}^2(t) + Z_{k}^{\mathrm{C}}(t)(Q_{k}(t) - \bar{Q}_{k})]$ into expression \eqref{DriftPenaltyUpperBound}, it can be rewritten as
\begin{equation}\label{DriftPenaltyUpperBound2}\small{
\begin{split}
&\Delta_v (t)
\leq B_{1}
    - V_{0}\sum_{k \in \mathcal{K} }^{K} Z_{k}^{\mathrm{C}}(t) \mathbb{E}[\bar{Q}_{k}-Q_{k}(t)|\mathbf{\Theta}(t)]
  \\&  -  V_{1}I^{\mathrm{W}}(t) \mathbb{E}[\bar{I}^{\mathrm{W}}-I^{\mathrm{W}}(t)|\mathbf{\Theta}(t)]
   + V_{0}\sum_{k \in \mathcal{K}} \bigg[ Q_{k}^{2}(t) -
   \\& \left(Q_{k}(t)+Z_{k}^{\mathrm{C}}(t)\right)R_{k}(t)\tau(t) \bigg]
  - \mathbb{E} [R(t)|\mathbf{\Theta}(t)].
\end{split}}
\end{equation}
Since $(\bar{I}^{\mathrm{W}}-Z^{\mathrm{W}}(t))$ and $(\bar{Q}_{k}- Q_{k}(t))$ are two factors of
inequality constraints \eqref{ZWiGig} and \eqref{ZNR}, once these constraints are satisfied and the distribution of channel is i.i.d.,
\eqref{DriftPenaltyUpperBound2} can be simplified by applying $\omega$-only policy as follows \cite{Neely2010Stochastic}:
\begin{equation}\small{
\begin{split}
\Delta_v (t)
& \!\leq\! B_{1} \!+\! V_{0}\!\sum_{k \in \mathcal{K}} \! \!\left[ Q_{k}^{2}(t) \!-\! \!\left(Q_{k}(t)\!+\!Z_{k}^{\mathrm{C}}(t)\right)\!R_{k}(t)\tau(t) \right]\!
 \!-\!R^{\mathrm{opt}} \\
& \!\leq\! B_{1} \!+\! V_{0}\sum_{k\in\mathcal{K}}\left(Q_{k}^{\max}\right)^2 \!-\! R^{\mathrm{opt}} .
\end{split}}
\end{equation}
Summing from $0$ to $T-1$, and let $T \to \infty$, the above equation can be rearranged by
\begin{equation}\small{
\begin{split}
&\lim_{T\to+\infty}\frac{1}{T}\sum_{t=0}^{T-1}\mathbb{E}\left[R(t)\right]
\geq  R^{\mathrm{opt}} \!- B_{1} -\! V_{0}\sum_{k\in\mathcal{K}}(Q_{k}^{\max})^2
\\& + \! \lim_{\!T\!\to+\infty}\! \frac{1}{T} \mathbb{E}\!\left[\!V_0\sum_{k}L(Z_{k}^{\mathrm{C}}(t)) \!+\! V_1L(Z^{\mathrm{W}}(t))\right]\! \\
&\geq R^{\mathrm{opt}} -\! B_{1} -\! V_{0}\sum_{k\in\mathcal{K}}(Q_{k}^{\max})^2,
\end{split}}
\end{equation}
which ends the proof.


\ifCLASSOPTIONcaptionsoff
  \newpage
\fi

\end{document}